%% file: main.tex
\PassOptionsToPackage{dvipsnames}{xcolor}
\documentclass[sigplan,screen,nonacm]{acmart}
\renewcommand\footnotetextcopyrightpermission[1]{}
\AtBeginDocument{%
  }

\input{packages}

\input{macros}

\input{figures/exp_info}

\begin{document}

\title{\toolname{}: Opportunistically Exploiting Parallelism in Python Agentic Workflows}

\author{Stephen Mell}
\orcid{0009-0003-7469-8974}
\affiliation{%
  \institution{Carnegie Mellon University}
  \city{Pittsburgh}
  \state{Pennsylvania}
  \country{USA}
}
\email{stephenmell@cmu.edu}

\author{David Mell}
\orcid{0009-0008-0827-6120}
\affiliation{%
  \institution{Independent Researcher}
  \city{Juneau}
  \state{Alaska}
  \country{USA}
}
\email{zraexy@gmail.com}

\author{Harry Y. H. Li}
\orcid{0009-0005-8230-5282}
\affiliation{%
  \institution{University of Pennsylvania}
  \city{Philadelphia}
  \state{Pennsylvania}
  \country{USA}
}
\email{flandia@engineering.upenn.edu}

\author{Zhihao Jia}
\orcid{0000-0002-1270-5185}
\affiliation{%
  \institution{Carnegie Mellon University}
  \city{Pittsburgh}
  \state{Pennsylvania}
  \country{USA}
}
\email{zhihao@cmu.edu}

\author{Konstantinos Kallas}
\orcid{0000-0002-8984-6648}
\affiliation{%
  \institution{University of California, Los Angeles}
  \city{Los Angeles}
  \state{California}
  \country{USA}
}
\email{kkallas@ucla.edu}

\author{Steve Zdancewic}
\orcid{0000-0002-3516-1512}
\affiliation{%
  \institution{University of Pennsylvania}
  \city{Philadelphia}
  \state{Pennsylvania}
  \country{USA}
}
\email{stevez@engineering.upenn.edu}

\author{Osbert Bastani}
\orcid{0000-0001-9990-7566}
\affiliation{%
  \institution{University of Pennsylvania}
  \city{Philadelphia}
  \state{Pennsylvania}
  \country{USA}
}
\email{obastani@engineering.upenn.edu}

\renewcommand{\shortauthors}{Trovato et al.}

\begin{abstract}
\input{abstract}
\end{abstract}

\settopmatter{printfolios=true}
\maketitle
\pagestyle{plain}

\input{sec-intro}

\input{sec-motivating}

\input{sec-technical}

\input{sec-evaluation}

\input{sec-related-work}

\input{sec-conclusion}

\bibliographystyle{ACM-Reference-Format}
\bibliography{references}

\end{document}

%% file: packages.tex
\usepackage[frozencache=true]{minted}
\usepackage{subcaption}
\usepackage{xcolor}
\usepackage{balance}
\usepackage{tikz}
\usetikzlibrary{positioning, arrows.meta, shapes.geometric, calc, fit, tikzmark}
\usepackage{booktabs}
\usepackage{pifont}
\usepackage{makecell}
\usepackage{todonotes}

%% file: macros.tex
\newcommand{\sm}[1]{}
\newcommand{\sz}[1]{}
\newcommand{\smf}[1]{}

\newcommand{\hidesubmission}[1]{}

\newcommand{\toolname}[0]{\textsc{PopPy}}
\newcommand{\toolacronym}[0]{\textbf{P}arallel \textbf{Op}portunistic \textbf{Py}thon}
\newcommand{\pysubsetname}[0]{\toolname{}}
\newcommand{\opalname}[0]{$\lambda^O$}

\newcommand{\CAI}[0]{Agentic Workflow}
\newcommand{\CAIS}[0]{\CAI s}

\newmintinline[pyinline]{python}{}

\newcommand{\heading}[1]{\par\vspace{3pt}\noindent\textbf{#1}.\enspace}

\newcommand{\featyes}{\ding{51}}
\newcommand{\featno}{\ding{55}}
\newcommand{\featpart}{$\boldsymbol{\sim}$}

\newcommand{\benchsot}{ReportGenerator}
\newcommand{\benchtraq}{ReliableRagQA}
\newcommand{\benchdae}{PolicyAgents}
\newcommand{\benchtot}{MathSolver}
\newcommand{\benchbird}{MedicalAdvice}
\newcommand{\benchcamel}{SecAssistant}
\newcommand{\benchvipergpt}{ImageAnalysis}

\newcommand{\ExpEmpInternalMaxPerc}[0]{1\%}

%% file: figures/exp_info.tex
\newcommand{\ExpToTPyMeanTotalTime}[0]{271}
\newcommand{\ExpToTOurTotalTime}[0]{40}
\newcommand{\ExpToTOurSpeedupRatioOfMedians}[0]{6.7}
\newcommand{\ExpSetupCPU}[0]{24-core Intel(R) Xeon(R) Gold 6342 CPU @ 2.80GHz}
\newcommand{\ExpSetupRAM}[0]{1007GB RAM}
\newcommand{\ExpSetupOS}[0]{Debian GNU/Linux 11 (bullseye)}
\newcommand{\ExpSetupKernel}[0]{5.10.0-34-amd64}
\newcommand{\ExpSetupPython}[0]{3.13.7}
\newcommand{\ExpSetupOpenAI}[0]{2.16.0}
\newcommand{\ExpToolLoC}[0]{7465}
\newcommand{\ExpToolLoCCompiler}[0]{4162}
\newcommand{\ExpToolLoCRuntime}[0]{3303}
\newcommand{\ExpNumTrials}[0]{10}
\newcommand{\ExpScalingNumTrials}[0]{10}
\newcommand{\ExpCompilationTimeTaskMin}[0]{2.05}
\newcommand{\ExpCompilationTimeTaskMax}[0]{51.15}
\newcommand{\ExpScalingBIRDRelativeSpeedupMin}[0]{1.51}
\newcommand{\ExpScalingBIRDRelativeSpeedupMax}[0]{10.82}
\newcommand{\ExpScalingToTRelativeSpeedupMin}[0]{2.39}
\newcommand{\ExpScalingToTRelativeSpeedupMax}[0]{9.80}
\newcommand{\ExpAnnotationsNumMethods}[0]{336}
\newcommand{\ExpAnnotationsTypesDisplay}[0]{\texttt{bool}, \texttt{int}, \texttt{float}, \texttt{complex}, \texttt{str}, \texttt{bytes}, \texttt{tuple}, \texttt{frozenset}, \texttt{frozendict}, \texttt{date}, \texttt{time}, \texttt{datetime}, \texttt{timedelta}, \texttt{NoneType}, \texttt{type}, enums, Pydantic \texttt{BaseModel}s (frozen)}

\newcommand{\ExpPerfOverallCoreGeomean}[0]{3.9}

\newcommand{\ExpPerfOverallAllMax}[0]{6.7}
\newcommand{\ExpPerfOverallAllGeomean}[0]{1.8}

\newcommand{\ExpPerfSpeedupLanggraphCoreGeomean}[0]{2.0}
\newcommand{\ExpPerfSpeedupAdkCoreGeomean}[0]{2.1}
\newcommand{\ExpPerfSpeedupLlamaindexCoreGeomean}[0]{2.1}
\newcommand{\ExpPerfSpeedupLiqpCoreGeomean}[0]{2.9}
\newcommand{\ExpPerfSpeedupMafCoreGeomean}[0]{3.2}
\newcommand{\ExpPerfSpeedupHaystackCoreGeomean}[0]{2.9}
\newcommand{\ExpPerfSpeedupHamiltonCoreGeomean}[0]{2.7}

\newcommand{\ExpPerfOverheadPercentManualCoreGeomean}[0]{17.2}
\newcommand{\ExpPerfInternalOverheadFactorCoreGeomean}[0]{70}
\newcommand{\ExpPerfInternalOverheadFactorCoreMin}[0]{17}
\newcommand{\ExpPerfInternalOverheadFactorCoreMax}[0]{300}
\newcommand{\ExpToTSetupLLM}[0]{gpt-5.6-luna}

\newcommand{\ExpToTPyLLMCallTimeMax}[0]{29.07}
\newcommand{\ExpTotalProgramCount}[0]{35}

%% file: abstract.tex
Agentic workflows, which compose calls to ML models using a general-purpose programming language like Python, are widely used for a variety of user-facing tasks, from software engineering to enterprise automation, making their end-to-end latency a critical bottleneck.
To improve their latency, developers can either manually parallelize them or use restricted agentic frameworks that cannot express or exploit all forms of parallelization.
We address this by developing \toolname{}, a system that parallelizes external calls, e.g., AI models, in agentic workflows written in a very expressive subset of Python.
\toolname{} combines an ahead-of-time compiler with a runtime, addressing three key challenges in extracting parallelism from Python applications: language complexity, dynamic dispatch, and variable mutation.
On a set of real-world agentic workflows, \toolname{} achieves up to $\ExpPerfOverallAllMax{}\times$ speedups in end-to-end execution time, compared to standard Python execution and several agentic frameworks, while preserving the sequential program semantics.

%% file: sec-intro.tex
\section{Introduction}
\label{sec:introduction}

Machine learning (ML) models, especially large language models (LLMs), are increasingly deployed not as standalone models, but as components of larger \emph{agentic workflows}
that programmatically compose ML models with other components, usually written in a standard language like Python or JavaScript.
Prominent examples include chatbots with retrieval-augmented generation~\citep{rag}, which combine LLMs with knowledge-base lookup; coding agents~\citep{sweagent}, which use shell utilities and spawn subagents; and AI mathematics systems, which leverage LLM-guidance for proof search~\citep{alphageometry,funsearch}.

These workflows have long end-to-end execution times due to reliance on heavy external calls to ML models, and therefore optimizing them with techniques like parallelization becomes crucial.
However, doing it manually is hard
since it requires careful use of constructs like multithreading and async/await;
for many developers, particularly data scientists and ML researchers, using these constructs at all can be prohibitive, with only $22\%$ of data scientists and less than $1\%$ of Python files on GitHub using async~\citep{jetbrains2024,yang2022complex-python}.
Even for experts, mixing concurrency with application logic significantly increases complexity, reduces maintainability, and produces bugs~\citep{lu2008learning,asadollah2017concurrency}.
While coding agents are rapidly replacing humans in the writing of code, they rely heavily on feedback from tests to validate their code, which are notoriously tricky to use for concurrency, and, like humans, they have a much harder time correctly writing and maintaining concurrent code than sequential~\citep{concur-benchmark,jain2025assessing,shao2026linguistic}.%

To address this, there has been significant development of programming frameworks, such as LangGraph~\citep{langgraph} and Google ADK~\citep{adk}, that express applications as graphs and can then optimize applications by exposing parallelism across independent nodes.
However, such frameworks, while sufficient for simple workflows, cannot parallelize (or even sometimes express) general control flow (\S\ref{sec:motivating:cais}), such as loops, which are needed for many agentic workflows, such as LLM-guided search~\citep{treeofthoughts,funsearch} and multi-agent systems~\citep{anthropic2025multiagent,dae}.

So what can be done? The ideal solution would be to correctly and automatically parallelize agentic workflows written in standard languages as sequential programs, 
avoiding maintainability issues and the risk of concurrency bugs.
Of course, automatic parallelization is an old idea~\citep{dennis1974dataflow,sisal,apostolakis2020perspective} but has proven extremely challenging in practice, in particular due to its instrumentation and other overheads that often outweigh the benefits.
In this paper, we make a key observation that agentic workflows are a qualitatively different workload from the programs targeted by
prior parallelization work:
they tend to be written in slow, high-level languages, and
they are \emph{externally bottlenecked},
spending almost all of their execution time 
waiting for external components---namely AI model calls implemented as REST API calls or GPU code.
Critically, where the overhead of parallelization may have previously been unacceptably high, for agentic workflows even high overhead on the Python code is insignificant compared to the benefit of parallelizing slow model calls. This suggests \emph{opportunistically} executing these applications: evaluating the Python code eagerly and out-of-order to discover and dispatch slow AI model calls as early as possible.

Building on this observation, we propose \toolname{}%
\footnote{\toolname{}: \toolacronym{}}
, a system that automatically parallelizes agentic workflows written as sequential Python code, yielding performance close to that of manually parallelized implementations.
Applications are written in an expressive subset of Python, requiring minimal changes on the part of their developers. 
At runtime, \toolname{} executes the Python code eagerly and out-of-order, dispatching independent external calls as early as possible so that they run in parallel. Order-sensitive external calls, such as printing, are still run sequentially, ensuring that the program output is equivalent to that of standard Python execution.
\toolname{} provides a built-in library of oft-used external functions, including remote and local LLMs, stateless HTTP requests, system IO, and embedding and vision models. This library is extensible, with the only requirement being that newly added functions are
annotated with their parallelizability properties.

To achieve its goals, \toolname{} needs to address three key challenges related to the Python language:
(1) language complexity, (2) dynamic dispatch (which obscures which call is being made) and (3) variable mutation.
Specifically, \toolname{} solves them by (1) compiling Python to a minimal language \opalname{}~\citep{opal} that supports opportunistic, out-of-order execution; (2) delegating parallelization decisions to runtime controllers; and (3) optimizing common variable mutation patterns.

We evaluate \toolname{} on a variety of agentic workflows from the literature and show that it improves execution time across standard Python execution and several state-of-the-art agentic frameworks (up to $\ExpPerfOverallAllMax{}\times$) 
while preserving the sequential program semantics.

In summary, \toolname{} contributes the following:%
\begin{enumerate}
\item A two-phase compiler that handles Python's complexity by transpiling all Python features to a core calculus $\lambda^O$~\citep{opal}, which can uncover parallelization opportunities between external calls at runtime (\S\ref{sec:compiler}).
\item A library of annotated external calls and a dynamic concurrency control subsystem that enable determining which calls can run in parallel at runtime in the presence of dynamic dispatch (\S\ref{sec:concurrency-control}).
\item Optimizations for key variable mutation patterns, enabling parallelism without changing the behavior of the original application (\S\ref{sec:variable-mutation}).
\end{enumerate}
Before diving into the key technical contributions, we show a motivating example~(\S\ref{sec:motivating}), provide an overview of \toolname{}~(\S\ref{sec:system}), and describe its user interface in detail~(\S\ref{sec:interface}).
After the technical contributions we evaluate \toolname{}~(\S\ref{sec:evaluation}), and provide a discussion of related work~(\S\ref{sec:related-work}).

%% file: sec-motivating.tex
\section{Motivating Example}
\label{sec:motivating}

We first introduce agentic workflows and their performance limitations using a characteristic example~(\S\ref{sec:motivating:cais}).
We then describe how \toolname{} can extract parallelism from this example, and a key insight that makes it possible~(\S\ref{sec:motivating:with-poppy}).

\subsection{\CAIS{}}
\label{sec:motivating:cais}

Modern AI systems combine model calls with traditional programmatic components to solve complex tasks:
chatbots like ChatGPT rely on web search and document retrieval to factually ground responses~\citep{rag}, mathematical and scientific discovery systems use AI models to guide conventional search~\citep{alphageometry,funsearch}, and coding agents like Claude Code and Codex spawn sub-agents and execute code actions~\citep{claudecode,yao2022react,codeact}. We broadly refer to such systems as \emph{agentic workflows}.

\begin{figure}[t]
    \raggedright
    \tikzmark{totmintedtop}
    \begin{minted}[linenos,xleftmargin=1.75em,fontsize=\footnotesize]{python}
@ocelot
def agentic_beam_search(task):
    states = ("",)
    for step in range(NUM_STEPS):
        new_states = tuple()
        for s in states:
            new_states += llm_get_proposals(task, s)
        values = get_values(task, new_states)
        states = topk(states, values, BEAM_WIDTH)
        print(states)
    return states
@ocelot
def get_values(task, states):
    value_cache = frozenset()
    values = tuple()
    for idx, state in enumerate(states):
        if state in value_cache:
            value = 0
            print(f"{idx}: duplicate")
        else:
            value = llm_get_value(task, state)
            value_cache |= {state}
            print(f"{idx}: {value}")
        values += (value,)
    return values
@ocelot
def llm_get_proposals(task, state): ...
@ocelot
def llm_get_value(task, state): ...
# In Ocelot library:
def llm(prompt): ...
    \end{minted}
    \tikzmark{totmintedbottom}%
    \begin{tikzpicture}[overlay, remember picture]
        \draw[red, line width=1pt]
            ($(pic cs:totmintedtop)!0.5!(pic cs:totmintedbottom) + (3em, 26pt)$) --
            ($(pic cs:totmintedtop)!0.5!(pic cs:totmintedbottom) + (3em, -90pt)$);
        \draw[orange, line width=1pt]
            ($(pic cs:totmintedtop)!0.5!(pic cs:totmintedbottom) + (4.6em, -4pt)$) --
            ($(pic cs:totmintedtop)!0.5!(pic cs:totmintedbottom) + (4.6em, -80pt)$);
        \draw[olive, line width=1pt]
            ($(pic cs:totmintedtop)!0.5!(pic cs:totmintedbottom) + (6.2em, -13pt)$) --
            ($(pic cs:totmintedtop)!0.5!(pic cs:totmintedbottom) + (6.2em, -30pt)$);
        \draw[CarnationPink, line width=1pt]
            ($(pic cs:totmintedtop)!0.5!(pic cs:totmintedbottom) + (6.2em, -44pt)$) --
            ($(pic cs:totmintedtop)!0.5!(pic cs:totmintedbottom) + (6.2em, -70pt)$);
    \end{tikzpicture}%
    \caption{Agentic arithmetic reasoning system~\citep{treeofthoughts}, implemented in Python. The \pyinline{@ocelot} lines are
    added in order to execute those functions through \toolname{}. L1--L29 are provided by the workflow developer, while L31 is part of \toolname{}'s AI component library. The colored bars indicate specific code blocks referenced in Figure~\ref{fig:motivating:execution}.
    }
    \label{fig:motivating:tot}
\end{figure}

\heading{Example: Agentic Arithmetic Reasoning}
\label{sec:motivating:tot}
Suppose you need to build a system that finds an arithmetic expression to satisfy given constraints. Tree-of-Thoughts~\citep{treeofthoughts} is an agentic search procedure that can solve this kind of challenge
(and has also been used for proof search~\citep{he2025logictree}, robot planning~\citep{xu2025embodiedtreethoughtsdeliberate}, and enterprise modeling~\citep{totdomainmodeling}).

Figure~\ref{fig:motivating:tot} shows a simplified implementation in \toolname{}-supported Python.
\hidesubmission{}
It takes
constraints (\pyinline{task}) as input,
and then performs \pyinline{NUM_}\allowbreak\pyinline{STEPS} rounds of beam search.%
For each search state, i.e.,
arithmetic expression,
it gets successor states (\pyinline{llm_get_proposals}) using an LLM call,
scores them (\pyinline{get_}\allowbreak\pyinline{values}),
and then returns the top \pyinline{BEAM_WIDTH} candidates to continue the search.
\pyinline{get_values} scores states using another LLM call
i.e., how close the arithmetic expression is to satisfying the constraints,
uses a cache to avoid redundant LLM calls,  
and calls \pyinline{print} to log.
The functions \pyinline{llm_}\allowbreak\pyinline{get_proposals} and \pyinline{llm_get_value} first format a prompt, then call \pyinline{llm}, which makes an HTTP request to a remote LLM API (e.g., GPT~\citep{gpt}, Claude~\citep{claude}, or Gemini~\citep{gemini}), and parse the result into the final output type.

\heading{Performance Limitations}
\label{sec:motivating:parallelization}
When executed with standard Python, the program in Figure~\ref{fig:motivating:tot} is slow: on an \ExpSetupCPU{} machine using \ExpToTSetupLLM{}, the end-to-end latency is on average $\ExpToTPyMeanTotalTime{}$ seconds, where the LLM calls dominate the execution time (each taking up to \ExpToTPyLLMCallTimeMax{}\,s).
A common approach to improve the performance of such programs is by parallelizing some of the calls: 
the \pyinline{llm_get_proposals} calls (L7) are independent of each other, the \pyinline{llm_get_value} calls (L21) are independent of each other, and the \pyinline{llm} call itself is stateless.
However, there is currently no great approach to achieve parallelization for this and other agentic workflows.
On one hand there are various agentic frameworks that can automatically parallelize but do not support arbitrary control flow, like data-dependent loops. On the other hand, developers can manually parallelize Python code using asyncio and multithreading, which reduces maintainability and is error-prone: e.g., naively converting the \pyinline{for} loop on L16 into a parallel map operation almost works, except it creates a race on \pyinline{value_cache}.
Table~\ref{fig:motivating:features} summarizes these alternatives.

\begin{table}[t]
  \centering
  \scriptsize
  \begin{tabular}{l|cccc|c}
    \toprule
    \textbf{Approach} & \makecell{\textbf{Seq}\\\textbf{Ctrl}} & \makecell{\textbf{Par}\\\textbf{Map}} & \makecell{\textbf{Par}\\\textbf{Loop}} & \textbf{Func} & \makecell{\textbf{Auto}\\\textbf{Par}} \\
    \midrule
    LangGraph~\citep{langgraph}, ADK~\citep{adk}, LI-W~\citep{llamaindex}
    & \featyes & \featyes  & \featno & \featno & \featyes \\
    \midrule
    MAF~\citep{maf}, LI-QP~\citep{llamaindex}, Haystack~\citep{haystack}
    & \featyes & \featpart & \featno & \featno & \featyes \\
    \midrule
    Hamilton~\citep{hamilton}
    & \featno  & \featyes  & \featno & \featno & \featyes \\
    \midrule
    Python with asyncio
    & \featyes  & \featyes  & \featyes  & \featyes  & \featno  \\
    \midrule
    \toolname{} & \featyes  & \featyes  & \featyes  & \featyes  & \featyes \\
    \bottomrule
  \end{tabular}
  \caption{A comparison of the features provided by languages and frameworks for implementing agentic workflows:
  ``Seq Ctrl'' is \pyinline{if}, \pyinline{for}, \pyinline{while}, etc.; ``Par Map'' is the ability to parallelize map operations; ``Par Loop'' is the ability to parallelize loops; ``Func'' is support for function definitions; ``Auto Par'' indicates that parallelization is automatic. ``\featpart{}'' indicates that parallel maps only work for statically-sized containers.
  }
  \label{fig:motivating:features}
\end{table}

\subsection{Parallelizing with \toolname{}}
\label{sec:motivating:with-poppy}

We develop a system called \toolname{} that can automatically parallelize agentic workflows written in an expressive subset of Python while preserving their sequential behavior. In our approach, developers write \emph{internal code} that performs orchestration and does not perform any side-effects; this code can make \emph{external calls} to functions in the \toolname{} library such as ML models, other remote APIs, file operations, and native code.
\toolname{} achieves parallelism by analyzing and instrumenting internal code and then executing independent external calls with maximal parallelism, while preserving the order of external calls that depend on one another (e.g., the \pyinline{print} calls in Fig.~\ref{fig:motivating:tot}).
To enable such parallelization without prohibitive overhead, \toolname{} makes a key insight about the bottlenecks of agentic workflows.

\begin{figure}[t]
    \centering
    \includegraphics[width=\columnwidth]{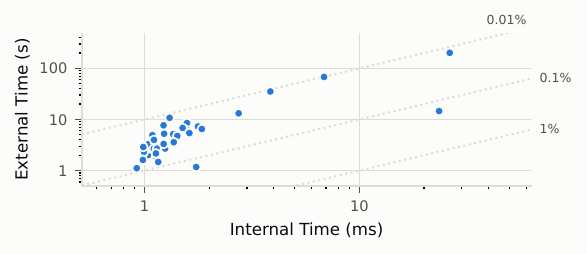}
    \caption{
        Time spent executing Python code (internal time) versus time spent waiting for AI model calls (external time), 
        for a set of 
        $\ExpTotalProgramCount{}$ agentic workflows (see \S\ref{sec:evaluation}). Dotted lines indicate constant fractions of total time spent in internal execution.
    }
    \label{fig:empirical-overhead}
\end{figure}

\heading{Internal vs External Bottlenecks}
\label{sec:motivating:external-vs-internal-performance}
The biggest part of the execution time of agentic workflows is spent in 
external calls
(a single LLM call can take seconds to complete)
that dominate the rest of the application code, which often simply glues together the calls and does not perform heavy computation.
To empirically understand the relation between external and internal execution time for agentic workflows, we measured the execution time for \ExpTotalProgramCount{} representative workflows collected from various sources in the literature (\S\ref{sec:benchmarks}), and separated it into internal and external execution time. Figure~\ref{fig:empirical-overhead} summarizes the results: all applications spend well under $\ExpEmpInternalMaxPerc{}$ of their execution time in internal code.
\toolname{} exploits this observation, adding significant overhead for internal Python code (through a custom interpreter and dynamic instrumentation), to parallelize the external code.

\begin{figure}[t]
    \centering
    \includegraphics[width=\columnwidth]{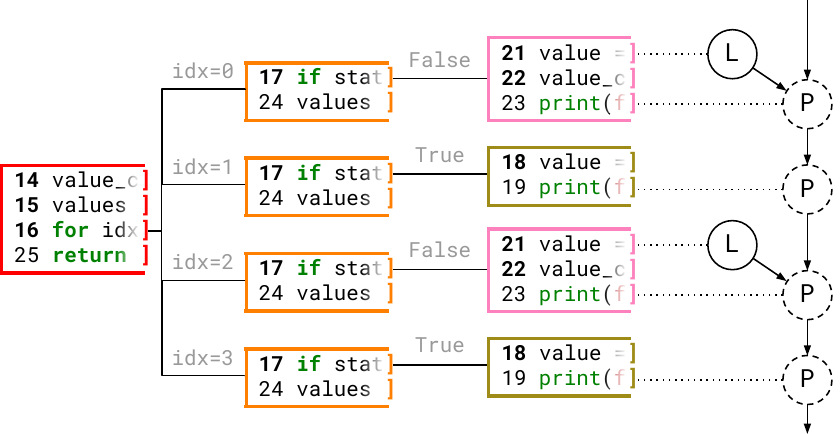}
    \caption{
        Illustration of the execution of \pyinline{get_values} with \pyinline{states = (a, a, b, b)}, after queueing all external calls but before any have resolved. \emph{Internal code execution tree} (left): Each code block that is executed at runtime is shown, including duplicates from different loop iterations; block borders are colored to correspond to source blocks in Figure~\ref{fig:motivating:tot}.
        Blocks for calls to \pyinline{llm_get_value} are omitted for space. \emph{External call dependency graph} (right): nodes represent queued external calls, and edges indicate dependencies that block execution; solid border means the call has been dispatched (L for \pyinline{llm} calls), dashed means it is waiting (``P'' for \pyinline{print}).
    }
    \label{fig:motivating:execution}
\end{figure}

\heading{Using \toolname{}}
To use \toolname{}, a developer typically writes internal code (decorated with \pyinline{@ocelot}) that orchestrates external functions in \toolname{}'s library. This library includes commonly used functions such as \pyinline{print} and \pyinline{llm}. Each external call is annotated so \toolname{} knows whether to run it in parallel or sequentially; e.g., \pyinline{print} calls have to be executed in sequence, while \pyinline{llm} calls are independent and can be executed in parallel. 
\toolname{} considers unannotated external functions as sequential, but developers can add appropriate annotations to parallelize them (see Sec.~\ref{sec:concurrency-control}).

\heading{Parallel Execution}
We describe how \toolname{} executes part of the code in Fig.~\ref{fig:motivating:tot}, in particular \pyinline{get_values} with 
\pyinline{states} set to abstract values \pyinline{(a,}\ \pyinline{a,}\ \pyinline{b,}\ \pyinline{b)}.
Figure~\ref{fig:motivating:execution} shows the tree of executed code blocks. 
Initially, execution proceeds as normal in the first loop iteration (for \pyinline{idx = 0}, \pyinline{state = a}), eventually calling \pyinline{llm_get_value} (L21, and top path in Fig.~\ref{fig:motivating:execution}), which itself makes an \pyinline{llm} call.
Instead of blocking on that call,
\toolname{} queues it, adds it to a dependency graph (right side of Fig.~\ref{fig:motivating:execution}), and assigns its result (\pyinline{value}) to a placeholder.
Because \pyinline{llm} is independent from other external calls 
and the call arguments are known, it is asynchronously dispatched in the ordinary Python interpreter.
Execution of internal code then continues, immediately executing L22 (since it doesn't depend on the external call result), and then queueing the call to \pyinline{print} (L23), which (in contrast to the \pyinline{llm} call) cannot dispatch immediately because it (1) depends on the unresolved \pyinline{value}, and (2) depends on all preceding sequential calls in order to preserve program behavior. 
Execution again continues, skipping internal code that depends on unresolved arguments (L24), proceeding to the second loop iteration and its \pyinline{if} branch, and queueing another \pyinline{print} call. 
After stepping through the final two loop iterations and queueing more external calls, execution is stuck, as all unevaluated statements depend on some outstanding external call. 
This state is depicted in Figure~\ref{fig:motivating:execution}.
When an external call returns (e.g., the remote LLM API returns a response), it \emph{resolves}, replacing its placeholder with the return value. This frees up dependent operations to execute. Suppose the first LLM call resolves, and there are no outstanding preceding sequential external calls. Then, the first \pyinline{print} call (L23) becomes unblocked and executes, as does the operation appending \pyinline{value} to \pyinline{values} (L24).

\heading{Speedups with \toolname{}}
Running this program 
with \toolname{} on the same setup, the end-to-end execution time reduces from $\ExpToTPyMeanTotalTime{}$ seconds to $\ExpToTOurTotalTime{}$ seconds, a $\ExpToTOurSpeedupRatioOfMedians{}\times$ speedup.

\section{System Overview}
\label{sec:system}

\begin{figure}[t]
    \centering
    \includegraphics[width=\columnwidth]{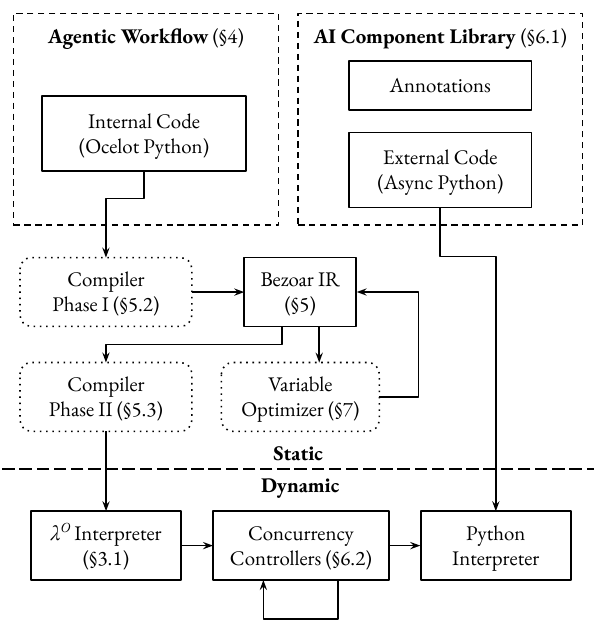}
    \caption{
        The \toolname{} system architecture. 
    }
    \label{fig:arch}
\end{figure}

Fig.~\ref{fig:arch} shows the architecture of \toolname{}, which includes both static and dynamic components. The user provides a Python program (like  Fig.~\ref{fig:motivating:tot}),
with the internal code decorated with \pyinline{@ocelot}.
Before execution,
the internal code is compiled to an intermediate representation called Bezoar~(\S\ref{sec:compiler}) which is in turn compiled to \opalname{}, a core calculus for extracting external-call parallelism.
A variable optimization pass is applied to the Bezoar code to improve parallelizability (\S\ref{sec:variable-mutation}).
The compiled internal code is executed by the \opalname{} interpreter, while the external code is executed with a standard Python interpreter. Whenever an external call is queued, a concurrency controller (\S\ref{sec:concurrency-control}) is spawned for it to coordinate execution order with other external calls.
\toolname{} is \emph{sound}~(\S\ref{sec:guarantees}),
i.e., applications have the same behavior as if they were executed with the standard Python interpreter,
which it achieves by executing dependent external calls in the sequential order.
In the rest of this section, we describe the key challenges to ``opportunistically'' executing Python code, and how the components of \toolname{} address these challenges.

\subsection{Extracting Parallelism with \opalname{}}
\label{ssec:lambdaO-background}
As illustrated in Section~\ref{sec:motivating:with-poppy}, \toolname{} leverages out-of-order execution to extract parallelism, continuing evaluation past parts of the program that are ordinarily blocked by outstanding external calls. To do this, \toolname{} relies on \opalname{}~\citep{opal}, a core calculus that crystallizes this idea, called \emph{opportunistic} evaluation.
Execution under opportunistic evaluation is highly nondeterministic, but guarantees \emph{confluence}---i.e., any evaluation order of \opalname{} matches sequential execution.
This allows \opalname{} to take advantage of parallelization opportunities that could not be leveraged statically---in different runs of the same program, external calls may have different durations and return in different orders; \opalname{} will always execute as much as possible, varying at runtime based on what calls have actually returned.
While \opalname{}'s interpreter obtains the desired parallelism, it is a functional, immutable, and non-sequential language, making it very different from Python, which is imperative, mutable, and sequential. Additionally, it assumes the behavior of external calls depends only on their arguments---there can be no hidden dependencies, which is not true of operations like \pyinline{print}. This discrepancy raises several challenges that \toolname{} must address.

\subsection{Challenge: Python Complexity (\S\ref{sec:compiler})}

\opalname{} is able to provide confluence for opportunistic evaluation in a provable manner by being extremely minimal:
like $\lambda$ calculus, its only language constructs are function definitions and function calls.
In contrast, Python is a complex language with many interacting features. While it is in principle possible to directly devise an opportunistic evaluation strategy for Python, it would be quite challenging.

Instead, \toolname{} relies on \opalname{}'s interpreter for the execution of internal code. To apply it to Python, \toolname{} compiles the internal Python code to \opalname{}. However, \opalname{} and Python have conflicting features (immutable vs. mutable, functional vs. imperative, parallel vs. sequential, etc.) raising several challenges. To bridge the gap, we split the compiler into two phases, with a novel intermediate representation (IR) between them called Bezoar (\S\ref{sec:compiler}). Like \opalname{}, Bezoar is minimal and explicit, but like Python, it is sequential and has variable mutation. The first phase converts complex features of Python into the minimal Bezoar language (\S\ref{sec:compiler:phase1}); the second phase converts the imperative (sequential, mutable) Bezoar into the functional (opportunistic, immutable) \opalname{} (\S\ref{sec:compiler:phase2}).

\subsection{Challenge: Dynamic Dispatch (\S\ref{sec:concurrency-control})}
To soundly extract parallelism, \toolname{} must execute sequential calls (e.g., \pyinline{print}) in order, while executing unordered calls (e.g., LLMs) early. In \opalname{}, sequential versus unordered behavior is determined at the call site: calls are only sequenced if they have a data dependency. Na\"{i}vely, our compiler could automatically introduce data dependencies to enforce sequencing where necessary; however, this strategy requires the reordering behavior of a call site to be known statically, which is not true for Python, since it supports dynamic dispatch.
For example, the \pyinline{+=} operator (L24) is dynamically dispatched based on the type of the left-hand side: for (immutable) tuples, it can be unordered, but for (mutable) lists, it must be sequential. Sequencing every dynamically dispatched call would be sound, but also removes almost all parallelism opportunities.

Instead, \toolname{} delegates concurrency control to the runtime. When an external call is queued, instead of having the \opalname{} interpreter launch the call directly, the interpreter launches a concurrency controller, a lightweight thread that wraps the external call. This controller knows what function is actually being called, and thus knows the desired concurrency behavior. The compiler encodes the original sequential order of calls into the \opalname{} code, which the concurrency controllers use to coordinate among themselves and determine when it is sound to actually dispatch their external call.

\subsection{Challenge: Variable Scoping and Mutation (\S\ref{sec:variable-mutation})}
\label{ssec:challenge-variable-mutation}

Concurrency and mutation are known to interact poorly. \opalname{} solves this by 
only supporting immutable variables, but
Python code makes heavy use of variable mutation. Because Python allows non-local uses of variables, mutations are hard to analyze and soundness, in the worst case, requires them to be strictly sequenced, preventing parallelization.

To address this challenge, \toolname{} makes the following observation: most variables in practical agentic workflows fall into one of two analyzable categories: (1) variables that are mutated but only used locally, and (2) variables that are used non-locally but only assigned once.
\toolname{} translates the assigned-once variables directly to $\lambda^O$ and implements a variable ``promotion'' pass that turns mutable local variable writes and reads into immutable variables. 
This avoids the need for unnecessary sequencing in most cases, enabling parallelization while retaining soundness.

%% file: sec-technical.tex
\section{Python Fragment and Correctness}
\label{sec:interface}

In this section, we describe \toolname{}'s supported Python fragment (\S\ref{sec:python-fragment}) and its correctness guarantee---it preserves the original sequential semantics modulo reorderings allowed by the annotated library (\S\ref{sec:guarantees}).

\subsection{Python Fragment for Internal Code}
\label{sec:python-fragment}

For internal (\pyinline{@ocelot}) code, \toolname{} supports an expressive subset of Python that can support many real-world applications,
including \pyinline{if} statements, \pyinline{for} and \pyinline{while} loops, function definitions and calls, variable assignments, tuples, operators, and non-local control flow such as \pyinline{break}, \pyinline{continue}, and \pyinline{return}.
The main restriction of \toolname{} is exceptions, which can only be raised; handling must occur in external code, since exceptions introduce unpredictable control-flow---even the potential for an exception prohibits opportunistically evaluating past that point.
\toolname{} wraps all external calls with a general exception handler; if it catches an exception, it terminates and issues an error to the user. Libraries can signal errors by returning enums or variants.

\subsection{Correctness Guarantee}
\label{sec:guarantees}

\toolname{} optimizes Python programs without affecting their observable behavior: the sequence of external calls is the same as that when running with standard Python, modulo the reorderings allowed by the annotations.%
Formally, we call a sequence of external calls a \emph{trace}. For external call annotations $A$, we define the equivalence relation over traces $\equiv_A$ such that $t_1 \equiv_A t_2$ iff $t_1$ can be transformed into $t_2$ via external call reorderings allowed by $A$.
We define the semantics $\llbracket P \rrbracket_S$ to be the trace produced by $P$ when executed with system $S$, and let $\operatorname{\toolname{}}(A)$ be \toolname{} with annotations $A$.
\begin{proposition}[Soundness]
For all programs $P$ and annotations $A$, $\llbracket P \rrbracket_{\operatorname{\toolname{}}(A)} \equiv_A \llbracket P \rrbracket_{\operatorname{Python}}$.
\end{proposition}
\begin{proof}[Proof sketch]
$\lambda^O$'s semantics are \textit{confluent}~\citep{opal}, 
i.e., the order in which internal computation is evaluated does not affect the final result. Thus, any $\lambda^O$ evaluation order will queue the same external calls as sequential $\lambda^O$ evaluation. The translation of Python into $\lambda^O$ via Bezoar introduces data dependencies that 
capture the sequential source semantics; thus, sequential execution in \toolname{} produces the same trace as Python execution. 
\toolname{}'s concurrency control may reorder external calls,
but only in ways that are $A$-equivalent to sequential \toolname{} execution. Therefore, \toolname{} execution is $A$-equivalent to Python execution.
\end{proof}

\section{Compiler}
\label{sec:compiler}
\toolname{} transforms internal Python into \opalname{}, two languages with very different characteristics:
(1) \emph{Sequentiality}: Python executes sequentially and has control-flow constructs like \pyinline{if} and \pyinline{for}; \opalname{} executes statements in arbitrary order and only has functions.
(2) \emph{Mutability}: Python variables are mutable whereas \opalname{}'s are immutable.
(3) \emph{Scoping}: Python's variable scoping is implicit, complex, and unusual~\citep{lambdapi};
\opalname{}'s is explicit and straightforward.
(4) \emph{Redundancy}: Python has many language features with overlapping semantics, while \opalname{} is a minimal calculus, simplifying analysis and execution.
This section describes  how \toolname{} bridges these characteristics with its two-phase compiler:
we first describe the Bezoar intermediate representation
and then the two phases, from Python to Bezoar (\S\ref{sec:compiler:phase1}) and Bezoar to \opalname{} (\S\ref{sec:compiler:phase2}).

\heading{Bezoar Intermediate Representation}
\label{sec:compiler:bezoar}
Bezoar sits in between Python and \opalname{}, sharing features with both languages: it retains the core semantic challenges of Python---se\-quen\-tial\-i\-ty and mutability---while making semantics explicit and only supporting minimal program constructs---if statements, for loops, function definitions, and calls. It has both mutable variables (\pyinline{x = ...}; like Python) and immutable variables (\pyinline{r := ...}; like \opalname{}).
Like Python, Bezoar executes sequentially and has control-flow constructs. However, the sequence of operations is made explicit through the use of A-normal form~\citep{anf}: Bezoar lacks nested expressions (e.g., \pyinline{z := g(f(x))}), and instead expresses each operation as a separate statement (e.g., \pyinline{y := f(x); z := g(y)}) to facilitate encoding statement sequencing as \opalname{} data-dependencies.

\subsection{Compiling \pysubsetname{} to Bezoar}
\label{sec:compiler:phase1}
The first phase consists of three steps:

\heading{Desugaring}
\toolname{} first directly desugars many Python features.
For example, object field access, \pyinline{x.y}, is replaced by \pyinline{getattr(x, "y")} and indexing, \pyinline{x[y]}, by \pyinline{getitem(x, y)}. 
This includes operators: 
\pyinline{x += y} becomes \pyinline{x = iadd(x, y)}.

\heading{Variable Scope Elaboration}
Python scoping is implicit, i.e., variable declarations are determined by whether there exists an assignment in a specific scope. 
This leads to subtle scoping semantics, as also studied in prior work~\citep{politz2013python}.
\toolname{}'s first compilation pass
makes Python's scoping semantics and mutable variables explicit,
with distinct declaration, \texttt{load}, and \texttt{store} constructs,
facilitating later analysis and optimization of variable mutation (\S\ref{sec:variable-mutation}).

\heading{A-Normalization}
\toolname{} makes the order of execution explicit by unfolding nested expressions,
converting each subexpression to a Bezoar statement and binding the result to a new immutable variable:
e.g., {\small \pyinline{z = x.f(y)}} compiles to:
\begin{minted}[linenos,xleftmargin=1.75em,fontsize=\small]{python}
r0 := load x; r1 := getattr(r0, "f");
r2 := load y; r3 := r1(r2); store z r3
\end{minted}

\subsection{Bezoar to $\lambda^O$}
\label{sec:compiler:phase2}

The second phase of \toolname{}'s compiler transforms Bezoar programs to $\lambda^O$ programs---in particular, encoding Bezoar's mutation, sequencing, and control flow into \opalname{}, which is immutable, does not have a sequential execution order, 
and only supports functions. Figure~\ref{fig:motivating:bezoar2opal} shows an example.

\begin{figure}[t]
    \begin{subfigure}{0.39\columnwidth}
        \begin{minted}[linenos,xleftmargin=1.75em,fontsize=\footnotesize]{python}
store x "foo"
r1 := print("bar")
if r0:
    r2 := load x
    r3 := print(r1)

else:
    store x "baz"

        \end{minted}
        \caption{Bezoar code, input}
        \label{fig:motivating:bezoar2opal:input}
    \end{subfigure}
    \begin{subfigure}{0.59\columnwidth}
        \begin{minted}[linenos,xleftmargin=1.75em,fontsize=\footnotesize]{python}
M1 := store(M0, x, "foo")
S1, r1 := print(S0, "bar")
def _then():
    r2 := load(M1, x)
    S2, r3 := print(S1, r2)
    return M1, S2
def _else():
    M2 := store(M1, x, "baz")
    return M2, S1
M3, S3 := ite(r0, _then, _else)
        \end{minted}
        \caption{\opalname{} code, output}
        \label{fig:motivating:bezoar2opal:output}
    \end{subfigure}
    \caption{An example of the second compiler phase. Mutation, sequencing, and control-flow are converted to function calls. In \opalname{}, \pyinline{load}, \pyinline{store}, and \pyinline{ite} are external functions.}
    \label{fig:motivating:bezoar2opal}
\end{figure}

\heading{Variable Mutation}
The first key challenge that needs to be addressed is that $\lambda^O$ does not have a notion of state or mutable variables, while Bezoar and Python do.
To address this challenge, \toolname{} transforms all variable loads and stores to external calls and passes explicit ``memory variables'' \pyinline{M} between them, ensuring that all variable accesses respect sequential execution semantics~\citep{wadler-monads}.
For example, in Fig.~\ref{fig:motivating:bezoar2opal}, \pyinline{store x "foo"} (L1) is compiled to \pyinline{M1 := store(M0, x,}\allowbreak\pyinline{"foo")} (L1). The immutable variables \pyinline{M0} and \pyinline{M1} are dictionaries mapping (mutable) variable names to values, reflecting the program state during execution. The call to \pyinline{store} takes a memory state \pyinline{M0} and returns a new 
memory state, identical except that the key \pyinline{x} has value \pyinline{"foo"}.

\heading{Sequencing of External Calls}
The second challenge is that \opalname{} has no execution order---i.e., it can run code in parallel as long as data dependencies are respected. Thus, to encode the original, sequential order of external calls, we introduce ``sequence variables'' \pyinline{S}. Similar to memory variables, each external call receives the previous sequence variable as input and returns a new one as output.
In Fig.~\ref{fig:motivating:bezoar2opal}, the two \pyinline{print} calls are sequenced through their dependency via \pyinline{S1}.
\toolname{} also handles
Python \emph{object} mutability, e.g., \pyinline{x.f = y}, by treating Python's \pyinline{getattr} and \pyinline{setattr} functions as external calls and sequencing them in the same way.

\heading{Conditionals}
To address the lack of control flow constructs in $\lambda^O$, \toolname{} ``functionalizes'' control flow statements---i.e., it transforms them into function calls using Church encodings~\citep{churchencoding}.%
For \pyinline{if} statements, each branch gets compiled into a function. The two branches, along with the condition, are then passed to a function \pyinline{ite} (Fig.~\ref{fig:motivating:bezoar2opal:output}, L10) that calls the \pyinline{_then} function if the condition is true and the \pyinline{_else} function otherwise.
Note that the \pyinline{M} and \pyinline{S} variables also need to be passed through control flow to account for branches performing different memory accesses or external calls. As an example, the \pyinline{_then} branch in Fig.~\ref{fig:motivating:bezoar2opal:output} performs an external call, while the \pyinline{_else} branch modifies a mutable variable.

\heading{Control Transfer Statements}
The next challenge is to compile early \pyinline{return} statements, \pyinline{continue}, and \pyinline{break} while keeping $\lambda^O$ simple. \toolname{} first performs \textit{continuation pushing}: given a conditional followed by a local continuation, it recursively appends the continuation to paths that fall through, while paths ending in an abrupt control transfer bypass the continuation. The normal-flow rewrite is a statement-level analog of the conditional A-reduction used in A-normalization~\citep[Fig.~7]{anf}. This normalization places abrupt control transfers in tail position, but can increase code size and obscure parallelism by duplicating continuations inside conditional branches. \toolname{} therefore immediately hoists the longest common ($\alpha$-equivalent) statement suffix from branches. This is a structured form of \textit{tail merging}~\citep{fraser1984analyzing}, traditionally used to reduce code size; in \toolname{}, it also removes control dependencies obscuring parallelism, so the suffix can then be discovered and dispatched without awaiting the condition. Figure~\ref{fig:control-transfer} illustrates both passes.

\begin{figure}[t]
    \centering
    \begin{subfigure}[t]{0.32\columnwidth}
        \begin{minted}[linenos,xleftmargin=1.75em,fontsize=\footnotesize]{python}
if c: return v
if d: A
else: B
K

        \end{minted}
        \caption{Input}
        \label{fig:control-transfer:input}
    \end{subfigure}\hfill
    \begin{subfigure}[t]{0.32\columnwidth}
        \begin{minted}[linenos,xleftmargin=1.75em,fontsize=\footnotesize]{python}
if c: return v
else:
    if d: A; K
    else: B; K

        \end{minted}
        \caption{After pushing}
        \label{fig:control-transfer:pushed}
    \end{subfigure}\hfill
    \begin{subfigure}[t]{0.32\columnwidth}
        \begin{minted}[linenos,xleftmargin=1.75em,fontsize=\footnotesize]{python}
if c: return v
else:
    if d: A
    else: B
    K
        \end{minted}
        \caption{After merging}
        \label{fig:control-transfer:merged}
    \end{subfigure}
    \caption{Compilation of Control Transfer Statements, shown schematically in Bezoar. \pyinline{A} and \pyinline{B} are normally completing statement sequences; \pyinline{K} is the local continuation; \textit{continuation pushing} duplicates \pyinline{K} into both alternatives of \pyinline{d}, but not into the early-return path; \textit{tail merging} shares \pyinline{K} outside the conditional on \pyinline{d}, while preserving its dependence on \pyinline{c}.}
    \label{fig:control-transfer}
\end{figure}

\heading{Loops} 
The front end lowers \pyinline{for} and \pyinline{while} loops to a common Bezoar \pyinline{loop}. At the start of each iteration, \pyinline{for} advances its iterator and breaks on exhaustion, whereas \pyinline{while} reevaluates its condition and breaks when it is false. The loop body is then functionalized into a \pyinline{_body} function passed to the \pyinline{loop} primitive. \pyinline{continue} and \pyinline{break} become returns from \pyinline{_body}, carrying the updated state and corresponding \textit{continue} or \textit{break} signal; ordinary fall-through also returns a \textit{continue} signal. The primitive invokes \pyinline{_body} again with the updated state on \pyinline{continue} and returns the final state on \pyinline{break}. This is analogous to CPS encodings of loops as recursive continuations~\citep{kelsey1995correspondence}, but encapsulates repetition in a primitive.

\section{External Calls and Concurrency Control}
\label{sec:concurrency-control}

\toolname{} executes external calls in parallel while preserving the sequential program semantics. This poses two challenges.
First, it needs to know which external functions can be run in parallel, and which are dependent. This is obvious when two calls have a data dependency, i.e., the return value of one is the argument of another, but it is not always the case, e.g., \pyinline{print} calls need to be executed sequentially even though they do not have a clear data dependency.
Second, Python's dynamic dispatch prevents statically knowing which function will be invoked at each call site.
This section describes how \toolname{} addresses these challenges with an annotated library~(\S\ref{ssec:annotations-and-library}) and dynamic concurrency control~(\S\ref{ssec:dynamic-concurrency-control}).

\subsection{Annotations and Library}
\label{ssec:annotations-and-library}

\toolname{} allows external library functions to be decorated with \pyinline{@sequential}, \pyinline{@readonly}, and \pyinline{@unordered} to indicate when they can be run in parallel with other calls.
Parallelizable calls that take a long time should also be declared as \pyinline{async} to yield control to the interpreter.
To assist developers, it provides a library containing annotated \pyinline{async} implementations of several AI components, including LLMs, text embedding models, and vision models, as well as a generic asynchronous stateless HTTP method that can be used to invoke arbitrary ML models or stateless remote APIs.
It also contains annotations of many methods in the Python standard library: 
tt annotates all 28 unary and binary operators (e.g. \pyinline{+}, \pyinline{==}): if both arguments are immutable, the call is parallelizable; if one or both is mutable, the call is read-only (because prior mutations to the argument must be allowed to finish). It annotates all 13 in-place operators (e.g. \pyinline{+=}): if both arguments are immutable, it is parallelizable; if the right-hand side is mutable, it is read-only, and if the left-hand side is mutable, it is sequential. It also annotates all \ExpAnnotationsNumMethods{} methods of core immutable datatypes:\footnote{
\scriptsize
\ExpAnnotationsTypesDisplay{}
}
if all arguments are immutable, the call is parallelizable; otherwise, it is read-only.
If an annotation is missing for an external call, \toolname{} considers it to be \pyinline{@sequential}, to guarantee soundness.

\subsection{Dynamic Concurrency Control}
\label{ssec:dynamic-concurrency-control}

Recall that \opalname{} assumes that there is a data-dependency between external calls that should be run sequentially. However, Python's dynamic dispatch makes it impossible to know statically whether an external call is reorderable. For example, the \pyinline{+=} operator (L24) is dynamically dispatched based on the type of the left-hand side: for (immutable) tuples, it can be parallelized, but for (mutable) lists, it must be sequential. %
As a result, the compiler ensures soundness by introducing sequencing variables \pyinline{S} between all adjacent call sites (\S\ref{sec:compiler:phase2}). However, the original \opalname{} interpreter \emph{always} executes external calls with data-dependencies sequentially, so it would serialize \emph{all} external calls.

\heading{Queued external calls}
To address this issue, we introduce a new, \emph{queued} state for external calls, between when the call has been discovered by the \opalname{} interpreter and when all of its dependencies have resolved. As soon as a call is queued, we spawn a \emph{concurrency controller} for it that decides when to dispatch the call, possibly before preceding calls have finished. Because the controller operates at runtime, it knows which external function is actually being called, and thus its reordering annotation. To dispatch sequential and read-only calls, the controller must know when some or all preceding calls have resolved.

\heading{Controller communication}
The inputs and outputs to each external call are available to the controller as futures; input futures can be awaited, and output futures can be fulfilled. Because the sequencing variables \pyinline{S} are passed between adjacent call sites, we use them to communicate two key pieces of information between controllers: (1) a future $F_R$ (akin to a read lock), indicating whether all preceding ``sequential'' calls have resolved; and (2) a future $F_W$ (akin to a write lock), indicating whether all preceding ``sequential'' \emph{and} ``read-only'' calls have resolved.

\heading{Implementation}
Calls that are \pyinline{@sequential} wait for both $F_W$ and $F_R$ to resolve before dispatching, guaranteeing that no other sequential or read-only call is executing while they do.
After the call resolves, the controller fulfills the output futures $F_R'$ and $F_W'$, notifying subsequent external calls that they can begin.
Calls that are \pyinline{@readonly} only wait for $F_R$, allowing them to run without waiting for other read-only calls. Once $F_R$ resolves, $F_R'$ is resolved (forwarding to subsequent calls the fact that prior sequential calls have completed), and then the call is dispatched. After the call resolves, the controller waits for $F_W$. Finally, it resolves $F_W'$ (indicating that this and all previous calls are complete).
Calls that are \pyinline{@unordered} do not wait for either future, dispatching immediately and forwarding the input $F_W$, $F_R$ as their output $F_W'$, $F_R'$.
Implementing concurrency control this way, by passing ``locks'' through the sequence variables, is extensible: finer-grained reorderability can be added in the future by passing finer-grained locks.

\section{Optimizing Variable Mutation}
\label{sec:variable-mutation}

While reads of immutable state are always safe to execute concurrently, state mutation is not, and handling concurrent mutations is known to be challenging.
By default, \toolname{} is conservative, matching sequential execution by ordering variable access with the memory variables from Section~\ref{ssec:dynamic-concurrency-control}. While sound, serializing all variable operations in this way blocks almost all of the parallelism we want to extract. Though serialization is necessary in the worst case, we observe that in practice, most variables fall into one of two patterns that we can handle: (1) only assigned once, or (2) only accessed locally. We provide optimizations for each of these cases that avoid serialization and allow \toolname{} to unlock parallelism. Additionally, these optimizations allow variables to be omitted from the global memory states \pyinline{M}, improving performance and lowering overhead.

\heading{Single-assignment variables}
If a mutable variable is only assigned once, then \toolname{} turns it into an immutable variable, statically resolving its read operations without the need for the global memory object. To illustrate the importance of this optimization, note that Python allows library functions such as \pyinline{print} to be reassigned. Na\"{i}vely, this means that all call sites require the function to be loaded from a variable, causing everything to be serialized. However, reassigning library functions is rare in practice, and this optimization avoids serialization in most cases.

\heading{Local variable promotion}
If a variable is only written and accessed in a local scope, \toolname{} uses a variable promotion SSA transformation~\citep{register-promotion} to unfold its loads and stores, avoiding the global memory object.
For straight-line code (no control-flow constructs), this essentially amounts to, for a mutable variable \pyinline{x} with $n$ store operations, splitting it into $n$ immutable variables \pyinline{x1} through \pyinline{xn}, replacing the $i$-th store operation \pyinline{store x r} with \pyinline{xi := r}, and replacing each load operation \pyinline{r := load x} with \pyinline{r := xj}, where $j$ is the index of the nearest preceding store operation.
To handle control-flow, we mirror the way the memory (\pyinline{M}) and sequence (\pyinline{S}) variables are passed through the program. For each \pyinline{if} statement and \pyinline{for} loop, we statically determine which variables occur in \pyinline{load} and \pyinline{store} operations. These variables are included alongside the memory and sequence variables, being returned from conditional branches and loop bodies and reassigned by the \pyinline{ite} and \pyinline{loop} call sites.

%% file: sec-evaluation.tex
\section{Evaluation}
\label{sec:evaluation}

\begin{figure*}[t]
    \centering
    \begin{subfigure}[c]{0.66\textwidth}
        \centering
        \includegraphics[width=\linewidth]{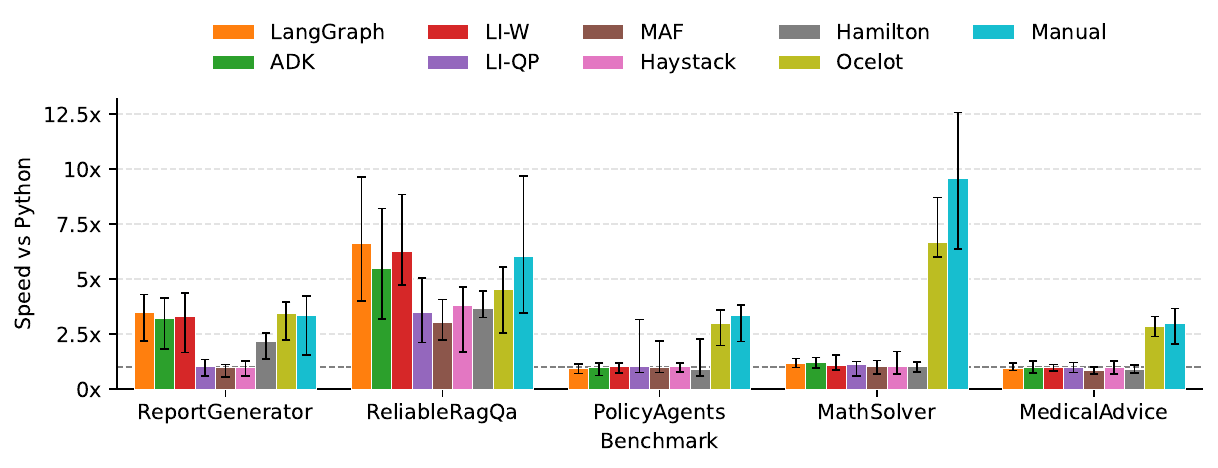}
    \end{subfigure}
    \hfill
    \begin{subfigure}[c]{0.3\textwidth}
        \centering
        \includegraphics[width=\linewidth]{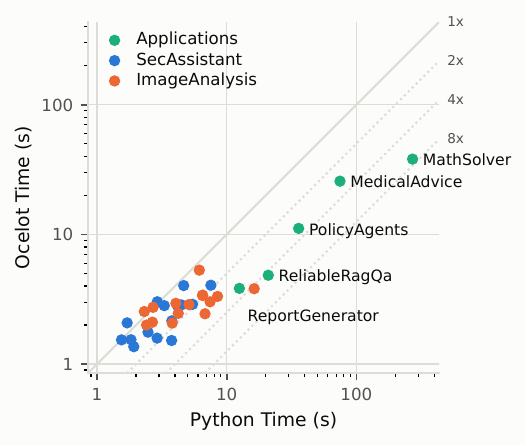}
    \end{subfigure}
    \caption{
        \textbf{Left:} median speedups over standard Python (min/max errorbars).
        \textbf{Right:} median Python and \toolname{} execution times for all workflows. Diagonal lines indicate speedup factors. Execution times are across \ExpNumTrials{} trials.
    }
    \label{fig:performance}
\end{figure*}

We evaluate \toolname{} on a variety of agentic workflows drawn from the literature (\S\ref{sec:benchmarks}).
Our evaluation answers the following questions:
(1) what is the performance of \toolname{} compared to state-of-the-art baselines~(\S\ref{sec:eval-performance});
(2) what are \toolname{}'s overheads~(\S\ref{sec:eval-overheads});
and (3) can \toolname{} scale with increasing parallelism potential~(\S\ref{sec:eval-scaling}).

\heading{Implementation}
We have implemented a prototype of \toolname{} in \ExpToolLoC{} lines of Python code.
The implementation comprises \ExpToolLoCCompiler{} LoC for the compiler and \ExpToolLoCRuntime{} LoC for the concurrency controllers and annotations. LLM coding agents were used to assist in its development.

\heading{Setup}
All experiments were conducted on a 
machine
with \ExpSetupCPU{} and \ExpSetupRAM{}. The OS is \ExpSetupOS{} with kernel version \ExpSetupKernel{}, Python version \ExpSetupPython{}, and OpenAI Python SDK version \ExpSetupOpenAI{}. LLM calls used
\ExpToTSetupLLM{}.
In all experiments, we execute applications \ExpNumTrials{} times.
Because the OpenAI API occasionally hangs for long periods of time, leading to outliers, we report medians across trials.

\heading{Baselines}
We compare \toolname{} to 
ordinary sequential Python execution and
several agent frameworks:  
LangGraph~\citep{langgraph}, 
Google ADK~\citep{adk}, 
Microsoft Agent Framework~\citep{maf} (MAF),
LlamaIndex~\citep{llamaindex} Workflows (LI-W),
LlamaIndex Query Pipelines (LI-QP), 
Haystack~\citep{haystack}, 
and Hamilton~\citep{hamilton}. 
For the frameworks, we reimplement each program preserving its structure, extracting parallelism through the frameworks' graph structure, without using manual asyncio.
Finally, we also manually rewrite all programs using Python's asyncio to achieve maximal parallelism (``Manual'').
All implementations of each benchmark are checked to have the same semantics, including model calls, side-effects, and output.

\subsection{Benchmarks}
\label{sec:benchmarks}

\begin{table}[t]
\centering
\small
\caption{Benchmark program characteristics. \textbf{LoC} is the number of lines, \textbf{For} and \textbf{If} are the number of each construct, \textbf{Dyn} is the number of dynamically dispatched call sites (operators or methods), and \textbf{Ext} is the number of external functions used.
The two last lines are suites, each consisting of 15 workflows
and ranges are min--max across them.
}
\label{tab:benchmarks_overview}
\input{figures/benchmark-table}
\end{table}

We collected a set of agentic workflows 
from the literature. We used the authors' implementation where available with minor syntactic adaptations. Tab.~\ref{tab:benchmarks_overview} shows an overview.\hidesubmission{}

\heading{\benchbird{}}
Answers medical questions with reliable uncertainties by having an LLM train a Bayesian model~\citep{bird}.

\heading{\benchdae{}}
Uses a multi-agent team that debates to answer questions while accounting for diverse perspectives~\citep{dae}.

\heading{\benchtot{}}
Solves arithmetic reasoning tasks with Tree-of-Thoughts, a form of LLM-guided beam search~\citep{treeofthoughts} (described in \S\ref{sec:motivating:tot}).

\heading{\benchsot{}}
Generates project reports by using a team of agents that hierarchically decomposes the task~\citep{skeletonofthought}.

\heading{\benchtraq{}}
Answers user questions with formal reliability guarantees using a corpus of documents~\citep{traq}.

\heading{\benchcamel{}}
Programs generated using CaMeL~\citep{camel}, an agentic system that dynamically generates an agentic workflow in a subset of Python to solve a given task. We drew tasks from the AgentDojo~Workspace~\citep{agentdojo},
restricting to the 15 workflows that do not raise runtime errors during ordinary Python execution and make at least one LLM call.

\heading{\benchvipergpt{}}
Programs generated using ViperGPT~\citep{suris2023vipergpt},
a system that answers natural-language queries about images by translating the query into a Python agentic workflow that makes calls to specialized vision models. 
We generated programs for 15 tasks randomly selected from the GQA~\citep{gqa} image question-answering benchmark.

\subsection{Overall Performance}
\label{sec:eval-performance}

Figure~\ref{fig:performance} shows the results: on the left the speedups of all the baselines compared to standard Python execution on the 5 human-written agentic workflows, and on the right the execution time of \toolname{} and Python on all workflows.
\toolname{} has a geomean speedup of $\ExpPerfOverallAllGeomean{}\times$ across all workflows over Python. 
For the 5 agentic workflows on the left, \toolname{} has geomean speedups of $\ExpPerfOverallCoreGeomean{}\times$ over Python, 
and $\ExpPerfSpeedupLanggraphCoreGeomean{}\times$,
$\ExpPerfSpeedupAdkCoreGeomean{}\times$,
$\ExpPerfSpeedupLlamaindexCoreGeomean{}\times$,
$\ExpPerfSpeedupLiqpCoreGeomean{}\times$,
$\ExpPerfSpeedupMafCoreGeomean{}\times$,
$\ExpPerfSpeedupHaystackCoreGeomean{}\times$,
$\ExpPerfSpeedupHamiltonCoreGeomean{}\times$ over LangGraph, ADK, LI-W, LI-QP, MAF, Haystack, and Hamilton respectively.
\toolname{} is $\ExpPerfOverheadPercentManualCoreGeomean{}\%$ slower than the manually parallelized implementation.

\heading{Discussion}
\toolname{} can successfully improve the end-to-end execution time compared to Python for all benchmarks by parallelizing external calls. 
It is also always able to express as much parallelism as all frameworks (\benchsot{}, \benchtraq{}), or more (\benchdae{}, \benchbird{}, \benchtot{}), 
by exposing parallelization in general loops.
\toolname{} performs slightly worse than LangGraph, ADK, and LI-W on \benchtraq{} due to its interpreter and instrumentation overheads.

\hidesubmission{}

\heading{Case-study: \benchtot{}}
To better understand the performance of \toolname{}, we zoom in on the precise timeline of external calls in an execution of \benchtot{} with 2 steps of search and a beam size of 3 (Fig.~\ref{fig:ext-trace}).
The execution goes through four distinct phases.
The first LLM call corresponds to \pyinline{llm_get_}\allowbreak\pyinline{proposals} for the initial search state.
Then, many calls to \pyinline{llm_get_value} execute in parallel, after which a
\pyinline{print} call is unblocked from the results of these calls and executes.
The process repeats, executing three parallel \pyinline{llm_get_proposals} calls (due to the beam size) and then many \pyinline{llm_get_value} calls. A second \pyinline{print} is then able to execute, concluding the second step of search. The whole program concludes with a final \pyinline{print} call.

\begin{figure}[t]
    \centering
    \includegraphics[width=\columnwidth]{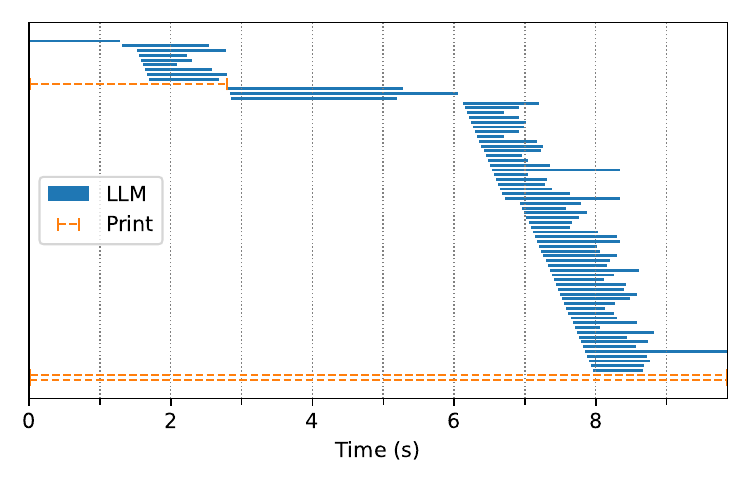}
    \caption{
        Execution trace of \benchtot{} (2 search steps and 3 beam size), showing selected external calls. LLM calls dispatch immediately and so solid lines indicate time between dispatch and resolution.
        Print calls resolve immediately so dashed lines indicate time between queueing and dispatch. 
        Calls are sorted 
        based on their sequential execution order.
    }
    \label{fig:ext-trace}
\end{figure}

\subsection{Overheads}
\label{sec:eval-overheads}

\begin{table}[t]
    \centering
    \scriptsize
    \caption{
        The time spent (seconds, mean $\pm$ stdev across \ExpNumTrials{} trials) executing only internal Python code without any pending external calls (Int) and the total execution time (Tot), for both Python and \toolname{}.
    }
    \input{figures/empirical_overhead_tab.tex}
    \label{tab:overhead:int_ext}
\end{table}

In this section, we evaluate \toolname{}'s compilation time and its execution overhead when there is no parallelization.

\heading{Compilation}
\toolname{} compilation time for all applications ranges from \ExpCompilationTimeTaskMin{}ms to \ExpCompilationTimeTaskMax{}ms, making it fast enough to be used in the critical path when running such programs.

\heading{Execution}
Table~\ref{tab:overhead:int_ext} shows the amount of time that Python and \toolname{} spend executing internal code (when not blocked on an external call) and the total execution time for the human-written agentic workflows.
\toolname{} introduces significant internal overhead (geomean: $\ExpPerfInternalOverheadFactorCoreGeomean{}\times$, min: $\ExpPerfInternalOverheadFactorCoreMin{}\times$, max: $\ExpPerfInternalOverheadFactorCoreMax{}\times$). However, even so, the end-to-end execution time is significantly improved since it is dominated by external calls.

\subsection{Scaling}
\label{sec:eval-scaling}

\begin{figure}[t]
    \centering
    \includegraphics[width=\columnwidth]{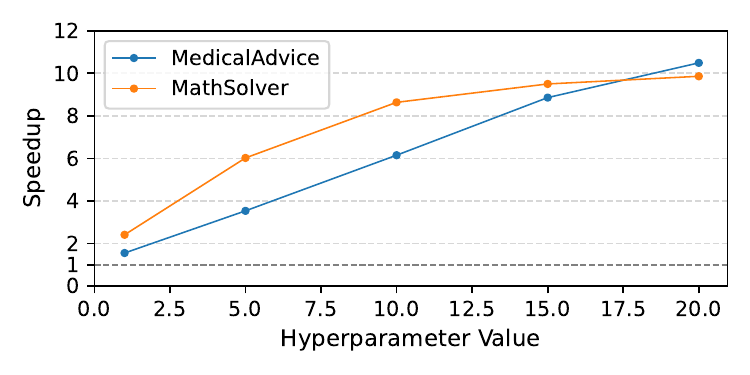}
    \caption{Speedup of \toolname{} over Python 
    (median over \ExpScalingNumTrials{} trials)
    as a function of a hyperparameter in each benchmark.
    }
    \label{fig:scaling:line}
\end{figure}

To evaluate how \toolname{} scales to more parallelization opportunities, we execute two programs (\benchbird{} and \benchtot{}) with hyperparameters affecting the number of LLM calls.
For \benchbird{}, we vary 
the number of times an LLM is called to assess a piece of evidence (originally 3), and for \benchtot{}, we vary \pyinline{BEAM_WIDTH} (Fig.~\ref{fig:motivating:tot} L9; originally set to 5), both from 1 to 20.
Figure~\ref{fig:scaling:line} shows results.
For \benchbird{} the speedup over Python scales from $\ExpScalingBIRDRelativeSpeedupMin{}\times$ to $\ExpScalingBIRDRelativeSpeedupMax{}\times$; for \benchtot{}, it scales from $\ExpScalingToTRelativeSpeedupMin{}\times$ to $\ExpScalingToTRelativeSpeedupMax{}\times$.

\heading{Discussion}
\toolname{} is able to leverage increasing parallelization potential in applications.
Note that the scaling is not linear because the parameters do not linearly lead to more parallelization, but just increase some LLM calls in the workflow---e.g., \benchbird{} has separate, parallelizable LLM calls for generating additional information about the task.
Furthermore, note that for \benchtot{} speedup is above 1 even for \pyinline{BEAM_WIDTH = 1} as the program has additional parallelism: \pyinline{llm_get_proposals} (Fig.~\ref{fig:motivating:tot}, L7) returns a tuple of states, which \pyinline{get_values} can loop over in parallel.

%% file: figures/benchmark-table.tex
\begin{tabular}{@{}lrrrrrr@{}}
    \toprule
    \textbf{Benchmark} & \textbf{\#} & \textbf{LoC} & \textbf{For} & \textbf{If} & \textbf{Dyn} & \textbf{Ext}  \\
    \midrule
    \benchbird{} & & 435 & 10 & 12 & 37 & 10 \\
    \benchdae{} & & 260 & 7 & 12 & 101 & 8 \\
    \benchtot{} & & 169 & 6 & 5 & 21 & 18 \\
    \benchsot{} & & 100 & 3 & 3 & 11 & 8 \\
    \benchtraq{} & & 210 & 3 & 2 & 5 & 9 \\
    \midrule
    \benchcamel{} & 15 & 2--114 & 0--7 & 0--17 & 0--53 & 0--8 \\
    \benchvipergpt{} & 15 & 13--24 & 1--2 & 2--5 & 4--13 & 4--5 \\
    \bottomrule
\end{tabular}

%% file: figures/empirical_overhead_tab.tex
\begin{tabular}{@{}lrrrr@{}}
    \toprule
    & \multicolumn{2}{c}{\textbf{Python}} & \multicolumn{2}{c}{\textbf{\toolname{}}} \\
    \cmidrule(lr){2-3} \cmidrule(lr){4-5}
    \textbf{Benchmark} & Int & Tot & Int & Tot \\
    \midrule
    \benchbird{} & $0.007 \pm 0.001$ & $73.3 \pm 9.3$  & $2.1 \pm 0.6$ & $27.2 \pm 4.4$ \\
    \benchdae{}  & $0.004 \pm 0.000$ & $35.4 \pm 7.9$  & $0.2 \pm 0.1$ & $11.8 \pm 3.3$ \\
    \benchtot{}  & $0.028 \pm 0.004$ & $268.7 \pm 69.6$ & $2.3 \pm 0.7$ & $37.6 \pm 12.8$ \\
    \benchsot{}  & $0.001 \pm 0.000$ & $12.3 \pm 2.0$  & $0.1 \pm 0.0$ & $3.9 \pm 1.0$ \\
    \benchtraq{} & $0.023 \pm 0.002$ & $20.5 \pm 5.4$  & $0.4 \pm 0.1$ & $4.5 \pm 2.1$ \\
    \bottomrule
\end{tabular}

%% file: sec-related-work.tex
\section{Related Work}
\label{sec:related-work}

There are three main lines of related work:
(1) specialized frameworks for parallelizing agentic workflows that have limited expressiveness or require manual parallelization;
(2) systems for parallelization of programs with external calls that cannot support imperative code or Python in particular;
and (3) systems for Python optimization that cannot extract parallelism across external calls.

\heading{Agent frameworks}
Numerous frameworks facilitate the development of agentic workflows, and they broadly support three abstractions: data-flow graphs, ``superstep'' control-flow graphs, and functional APIs with explicit futures. Some frameworks offer combinations of these abstractions and have shifted over time. These frameworks offer many features, from checkpointing and provider-agnostic APIs, to automatic parallelization, but in contrast to \toolname{}, all of them have limitations in terms of the amount of parallelism they can support, especially for applications with more complex and dynamic control flow.
In particular, data-flow graph based frameworks (e.g., LlamaIndex~\citep{llamaindex} Query Pipelines, Haystack~\citep{haystack}, and Apache Hamilton~\citep{hamilton}) cannot express control-flow constructs like loops;
control-flow graph frameworks (e.g., LangGraph~\citep{langgraph}, Google ADK~\citep{adk}, Microsoft Agent Framework~\citep{maf}) and LlamaIndex~\citep{llamaindex} Workflows can support loops but cannot extract parallelism from them;
and explicit functional APIs require manual parallelization with Python multithreading~\citep{python-threading} or asyncio~\citep{python-asyncio}.
In contrast, \toolname{} allows programmers the flexibility to write programs in standard Python while automatically parallelizing external calls.

\heading{Automatic parallelization of general-purpose languages}
There is an enormous literature on automatic parallelization~\citep{Karp1966,Adams1968,Karp1969,Adams1969,Rodriguez1969,dennis1974dataflow,davis1982data,mcgraw1982val,sisal,arvind1992id,maessen1995semantics,multilisp,dataflow-history,apostolakis2020perspective} of general-purpose languages.
While most of this existing work does not support extracting parallelism from
programs bottlenecked at external calls, there has been recent interest in doing so~\citep{posh,mozart,dish,ignis:pldi:2019}.
Most relatedly, PaSh~\citep{pash:eurosys:2021,pash-jit-osdi} focuses on shell scripts, extracting parallelism across black-box external components; however, PaSh does not provide solutions for the challenges of Python code addressed by \toolname{} (e.g., function definitions, scoping, variable mutation, and dynamic dispatch).
Other systems in this space have different goals from \toolname{}, though they also use external call annotations to specify key properties such as:
(1) providing distribution and offloading hints (e.g., in POSH~\citep{posh} and Ignis~\citep{ignis:pldi:2019}),
(2) describing the inputs and outputs of an external call (e.g., in PaSh~\citep{pash:eurosys:2021} and DiSh~\citep{dish} that focus on shell commands whose inputs and outputs cannot be directly identified), 
or (3) describing how to split the inputs of an external call to enable data parallelism (e.g., in Mozart~\citep{mozart} and PaSh~\citep{pash:eurosys:2021}).
Drawing inspiration from these systems, \toolname{} also uses annotations to specify important properties of the external code, though they describe the filesystem accesses and shardability of functions and commands, whereas we describe how function reorderability varies based on call arguments.
Finally, \opalname{}~\citep{opal} is a core calculus for automatically extracting external-call parallelism; however, it does not operate on real-world programs, and does not handle features like complex scoping, variable mutation, and dynamic dispatch. \toolname{} builds on this work by bridging the gap between \opalname{} and real-world Python code.

\heading{Python optimization}
Recent work improves the performance of Python programs by transforming and optimizing bottleneck fragments in them~\citep{tuplex,hipy,codon,weld,pytorch2,dias,jax,numba,appy}.
Their key focus is to identify small bottleneck fragments of the internal Python computation (e.g., UDFs that are used as part of a data analytics application~\citep{hipy,tuplex} or tensor computations~\citep{pytorch2,jax}), understand their semantics and transform them into a different representation (e.g., LLVM or XLA), and finally apply aggressive optimizations to them.
Furthermore, Python has long been supported by a general purpose JIT compiler ecosystem, including PyPy~\citep{pypy} and more recently CPython's JIT compiler~\citep{cpython-jit}.
All of this work identifies that Python is hard to precisely analyze ahead-of-time, so it combines static analysis passes with runtime components that leverage runtime execution information to make analysis more precise.
\toolname{} follows the same hybrid approach, combining an ahead-of-time component with a runtime component to deal with Python's dynamic features such as dynamic dispatch and mutation.
At the same time, \toolname{} has a very different focus from these systems and addresses different challenges: these systems try to optimize applications where the bottleneck is code written in Python and therefore make restrictive assumptions about it, while \toolname{} explicitly slows down Python internal code to speed up end-to-end execution with parallelization.

%% file: sec-conclusion.tex
\section{Conclusion}

We have proposed \toolname{}, a system focused on automatically parallelizing agentic workflows written in Python. Our key insight is that these workflows are bottlenecked by calls to external code such as ML models; as a consequence, effective parallelization relies on identifying external calls that can be run in parallel.
Our evaluation shows that \toolname{} improves performance for a variety of agentic workflows, over both standard Python and existing agentic frameworks.

%% file: main.bbl

\begin{thebibliography}{78}


\ifx \showCODEN    \undefined \def \showCODEN     #1{\unskip}     \fi
\ifx \showISBNx    \undefined \def \showISBNx     #1{\unskip}     \fi
\ifx \showISBNxiii \undefined \def \showISBNxiii  #1{\unskip}     \fi
\ifx \showISSN     \undefined \def \showISSN      #1{\unskip}     \fi
\ifx \showLCCN     \undefined \def \showLCCN      #1{\unskip}     \fi
\ifx \shownote     \undefined \def \shownote      #1{#1}          \fi
\ifx \showarticletitle \undefined \def \showarticletitle #1{#1}   \fi
\ifx \showURL      \undefined \def \showURL       {\relax}        \fi
\providecommand\bibfield[2]{#2}
\providecommand\bibinfo[2]{#2}
\providecommand\natexlab[1]{#1}
\providecommand\showeprint[2][]{arXiv:#2}

\bibitem[Abbaspour~Asadollah et~al\mbox{.}(2017)]%
        {asadollah2017concurrency}
\bibfield{author}{\bibinfo{person}{Sara Abbaspour~Asadollah}, \bibinfo{person}{Daniel Sundmark}, \bibinfo{person}{Sigrid Eldh}, {and} \bibinfo{person}{Hans Hansson}.} \bibinfo{year}{2017}\natexlab{}.
\newblock \showarticletitle{Concurrency Bugs in Open Source Software: A Case Study}.
\newblock \bibinfo{journal}{\emph{Journal of Internet Services and Applications}} \bibinfo{volume}{8}, \bibinfo{number}{1} (\bibinfo{year}{2017}), \bibinfo{pages}{4}.
\newblock
\href{https://doi.org/10.1186/s13174-017-0055-2}{doi:\nolinkurl{10.1186/s13174-017-0055-2}}


\bibitem[Adams(1968)]%
        {Adams1968}
\bibfield{author}{\bibinfo{person}{Duane~A. Adams}.} \bibinfo{year}{1968}\natexlab{}.
\newblock \bibinfo{booktitle}{\emph{{A Computation Model with Data-Sequenced Control}}}.
\newblock \bibinfo{type}{{T}echnical {R}eport}. \bibinfo{institution}{Stanford University}.
\newblock
\shownote{Technical Report CGTM 45}.
\newblock


\bibitem[Adams(1969)]%
        {Adams1969}
\bibfield{author}{\bibinfo{person}{Duane~A. Adams}.} \bibinfo{year}{1969}\natexlab{}.
\newblock \emph{\bibinfo{title}{{A Computation Model with Data Flow Sequencing}}}.
\newblock \bibinfo{thesistype}{Ph.\,D. Dissertation}.
\newblock


\bibitem[Ansel et~al\mbox{.}(2024)]%
        {pytorch2}
\bibfield{author}{\bibinfo{person}{Jason Ansel}, \bibinfo{person}{Edward Yang}, \bibinfo{person}{Horace He}, \bibinfo{person}{Natalia Gimelshein}, \bibinfo{person}{Animesh Jain}, \bibinfo{person}{Michael Voznesensky}, \bibinfo{person}{Bin Bao}, \bibinfo{person}{Peter Bell}, \bibinfo{person}{David Berard}, \bibinfo{person}{Evgeni Burovski}, {et~al\mbox{.}}} \bibinfo{year}{2024}\natexlab{}.
\newblock \showarticletitle{Pytorch 2: Faster machine learning through dynamic python bytecode transformation and graph compilation}. In \bibinfo{booktitle}{\emph{Proceedings of the 29th ACM international conference on architectural support for programming languages and operating systems, volume 2}}. \bibinfo{pages}{929--947}.
\newblock


\bibitem[{Anthropic}(2024)]%
        {claude}
\bibfield{author}{\bibinfo{person}{{Anthropic}}.} \bibinfo{year}{2024}\natexlab{}.
\newblock \bibinfo{title}{The {Claude} 3 Model Family: Opus, Sonnet, Haiku}.
\newblock \bibinfo{howpublished}{\url{https://www-cdn.anthropic.com/de8ba9b01c9ab7cbabf5c33b80b7bbc618857627/Model_Card_Claude_3.pdf}}.
\newblock


\bibitem[{Anthropic}(2025a)]%
        {claudecode}
\bibfield{author}{\bibinfo{person}{{Anthropic}}.} \bibinfo{year}{2025}\natexlab{a}.
\newblock \bibinfo{booktitle}{\emph{{Claude Code}: An Agentic Coding Tool That Lives in Your Terminal}}.
\newblock
\urldef\tempurl%
\url{https://github.com/anthropics/claude-code}
\showURL{%
\tempurl}


\bibitem[{Anthropic}(2025b)]%
        {anthropic2025multiagent}
\bibfield{author}{\bibinfo{person}{{Anthropic}}.} \bibinfo{year}{2025}\natexlab{b}.
\newblock \bibinfo{title}{How we built our multi-agent research system}.
\newblock \bibinfo{howpublished}{\url{https://www.anthropic.com/engineering/multi-agent-research-system}}.
\newblock
\shownote{Accessed: 2026-04-01}.
\newblock


\bibitem[Apostolakis et~al\mbox{.}(2020)]%
        {apostolakis2020perspective}
\bibfield{author}{\bibinfo{person}{Sotiris Apostolakis}, \bibinfo{person}{Ziyang Xu}, \bibinfo{person}{Greg Chan}, \bibinfo{person}{Simone Campanoni}, {and} \bibinfo{person}{David~I August}.} \bibinfo{year}{2020}\natexlab{}.
\newblock \showarticletitle{Perspective: A sensible approach to speculative automatic parallelization}. In \bibinfo{booktitle}{\emph{Proceedings of the Twenty-Fifth International Conference on Architectural Support for Programming Languages and Operating Systems}}. \bibinfo{pages}{351--367}.
\newblock


\bibitem[Arvind(1992)]%
        {arvind1992id}
\bibfield{author}{\bibinfo{person}{Rishiyur S~Nikhil Arvind}.} \bibinfo{year}{1992}\natexlab{}.
\newblock \showarticletitle{Id: a language with implicit parallelism}.
\newblock In \bibinfo{booktitle}{\emph{A Comparative Study of Parallel Programming Languages}}. \bibinfo{publisher}{Elsevier}, \bibinfo{pages}{169--215}.
\newblock


\bibitem[Baziotis et~al\mbox{.}(2024)]%
        {dias}
\bibfield{author}{\bibinfo{person}{Stefanos Baziotis}, \bibinfo{person}{Daniel Kang}, {and} \bibinfo{person}{Charith Mendis}.} \bibinfo{year}{2024}\natexlab{}.
\newblock \showarticletitle{Dias: Dynamic rewriting of Pandas code}.
\newblock \bibinfo{journal}{\emph{Proceedings of the ACM on Management of Data}} \bibinfo{volume}{2}, \bibinfo{number}{1} (\bibinfo{year}{2024}), \bibinfo{pages}{1--27}.
\newblock


\bibitem[Bolz et~al\mbox{.}(2009)]%
        {pypy}
\bibfield{author}{\bibinfo{person}{Carl~Friedrich Bolz}, \bibinfo{person}{Antonio Cuni}, \bibinfo{person}{Maciej Fijalkowski}, {and} \bibinfo{person}{Armin Rigo}.} \bibinfo{year}{2009}\natexlab{}.
\newblock \showarticletitle{Tracing the Meta-Level: PyPy's Tracing JIT Compiler}. In \bibinfo{booktitle}{\emph{Proceedings of the 4th Workshop on the Implementation, Compilation, Optimization of Object-Oriented Languages and Programming Systems (ICOOOLPS)}}. \bibinfo{publisher}{ACM}, \bibinfo{pages}{18--25}.
\newblock
\href{https://doi.org/10.1145/1565824.1565827}{doi:\nolinkurl{10.1145/1565824.1565827}}


\bibitem[Bradbury et~al\mbox{.}(2021)]%
        {jax}
\bibfield{author}{\bibinfo{person}{James Bradbury}, \bibinfo{person}{Roy Frostig}, \bibinfo{person}{Peter Hawkins}, \bibinfo{person}{Matthew~James Johnson}, \bibinfo{person}{Chris Leary}, \bibinfo{person}{Dougal Maclaurin}, \bibinfo{person}{George Necula}, \bibinfo{person}{Adam Paszke}, \bibinfo{person}{Jake VanderPlas}, \bibinfo{person}{Skye Wanderman-Milne}, {et~al\mbox{.}}} \bibinfo{year}{2021}\natexlab{}.
\newblock \showarticletitle{Jax: Autograd and xla}.
\newblock \bibinfo{journal}{\emph{Astrophysics Source Code Library}} (\bibinfo{year}{2021}), \bibinfo{pages}{ascl--2111}.
\newblock


\bibitem[Bucher and Ostrowski(2024)]%
        {cpython-jit}
\bibfield{author}{\bibinfo{person}{Brandt Bucher} {and} \bibinfo{person}{Savannah Ostrowski}.} \bibinfo{year}{2024}\natexlab{}.
\newblock \bibinfo{title}{{PEP 744: JIT Compilation}}.
\newblock \bibinfo{howpublished}{\url{https://peps.python.org/pep-0744/}}.
\newblock
\shownote{Python Enhancement Proposal, Draft status}.
\newblock


\bibitem[Church(1941)]%
        {churchencoding}
\bibfield{author}{\bibinfo{person}{Alonzo Church}.} \bibinfo{year}{1941}\natexlab{}.
\newblock \bibinfo{booktitle}{\emph{The Calculi of Lambda-Conversion}}. \bibinfo{series}{Annals of Mathematics Studies}, Vol.~\bibinfo{volume}{6}.
\newblock \bibinfo{publisher}{Princeton University Press}.
\newblock


\bibitem[Davis and Keller(1982)]%
        {davis1982data}
\bibfield{author}{\bibinfo{person}{Alan~L. Davis} {and} \bibinfo{person}{Robert~M. Keller}.} \bibinfo{year}{1982}\natexlab{}.
\newblock \showarticletitle{Data Flow Program Graphs}.
\newblock \bibinfo{journal}{\emph{Computer}} \bibinfo{volume}{15}, \bibinfo{number}{02} (\bibinfo{date}{2} \bibinfo{year}{1982}), \bibinfo{pages}{26--41}.
\newblock


\bibitem[Debenedetti et~al\mbox{.}(2026)]%
        {camel}
\bibfield{author}{\bibinfo{person}{Edoardo Debenedetti}, \bibinfo{person}{Ilia Shumailov}, \bibinfo{person}{Tianqi Fan}, \bibinfo{person}{Jamie Hayes}, \bibinfo{person}{Nicholas Carlini}, \bibinfo{person}{Daniel Fabian}, \bibinfo{person}{Christoph Kern}, \bibinfo{person}{Chongyang Shi}, \bibinfo{person}{Andreas Terzis}, {and} \bibinfo{person}{Florian Tram{\`e}r}.} \bibinfo{year}{2026}\natexlab{}.
\newblock \showarticletitle{Defeating Prompt Injections by Design}. \bibinfo{howpublished}{arXiv preprint arXiv:2503.18813}. In \bibinfo{booktitle}{\emph{IEEE Conference on Secure and Trustworthy Machine Learning (SaTML)}}.
\newblock
\urldef\tempurl%
\url{https://arxiv. org/abs/2503.18813}
\showURL{%
\tempurl}


\bibitem[Debenedetti et~al\mbox{.}(2024)]%
        {agentdojo}
\bibfield{author}{\bibinfo{person}{Edoardo Debenedetti}, \bibinfo{person}{Jie Zhang}, \bibinfo{person}{Mislav Balunovic}, \bibinfo{person}{Luca Beurer-Kellner}, \bibinfo{person}{Marc Fischer}, {and} \bibinfo{person}{Florian Tram{\`e}r}.} \bibinfo{year}{2024}\natexlab{}.
\newblock \showarticletitle{Agentdojo: A dynamic environment to evaluate prompt injection attacks and defenses for llm agents}.
\newblock \bibinfo{journal}{\emph{Advances in Neural Information Processing Systems}}  \bibinfo{volume}{37} (\bibinfo{year}{2024}), \bibinfo{pages}{82895--82920}.
\newblock


\bibitem[Dennis(1974)]%
        {dennis1974dataflow}
\bibfield{author}{\bibinfo{person}{Jack~B. Dennis}.} \bibinfo{year}{1974}\natexlab{}.
\newblock \showarticletitle{First version of a data flow procedure language}. In \bibinfo{booktitle}{\emph{Programming Symposium}}, \bibfield{editor}{\bibinfo{person}{B.~Robinet}} (Ed.). \bibinfo{publisher}{Springer Berlin Heidelberg}, \bibinfo{address}{Berlin, Heidelberg}, \bibinfo{pages}{362--376}.
\newblock
\showISBNx{978-3-540-37819-8}


\bibitem[Feng et~al\mbox{.}(2025a)]%
        {dae}
\bibfield{author}{\bibinfo{person}{Yu Feng}, \bibinfo{person}{Phu~Mon Htut}, \bibinfo{person}{Zheng Qi}, \bibinfo{person}{Wei Xiao}, \bibinfo{person}{Manuel Mager}, \bibinfo{person}{Nikolaos Pappas}, \bibinfo{person}{Kishaloy Halder}, \bibinfo{person}{Yang Li}, \bibinfo{person}{Yassine Benajiba}, {and} \bibinfo{person}{Dan Roth}.} \bibinfo{year}{2025}\natexlab{a}.
\newblock \showarticletitle{Rethinking {LLM} Uncertainty: A Multi-Agent Approach to Estimating Black-Box Model Uncertainty}. In \bibinfo{booktitle}{\emph{Findings of the Association for Computational Linguistics: EMNLP 2025}}, \bibfield{editor}{\bibinfo{person}{Christos Christodoulopoulos}, \bibinfo{person}{Tanmoy Chakraborty}, \bibinfo{person}{Carolyn Rose}, {and} \bibinfo{person}{Violet Peng}} (Eds.). \bibinfo{publisher}{Association for Computational Linguistics}, \bibinfo{address}{Suzhou, China}, \bibinfo{pages}{12349--12375}.
\newblock
\showISBNx{979-8-89176-335-7}
\href{https://doi.org/10.18653/v1/2025.findings-emnlp.660}{doi:\nolinkurl{10.18653/v1/2025.findings-emnlp.660}}


\bibitem[Feng et~al\mbox{.}(2025b)]%
        {bird}
\bibfield{author}{\bibinfo{person}{Yu Feng}, \bibinfo{person}{Ben Zhou}, \bibinfo{person}{Weidong Lin}, {and} \bibinfo{person}{Dan Roth}.} \bibinfo{year}{2025}\natexlab{b}.
\newblock \showarticletitle{{BIRD}: A Trustworthy Bayesian Inference Framework for Large Language Models}. In \bibinfo{booktitle}{\emph{The Thirteenth International Conference on Learning Representations}}.
\newblock
\urldef\tempurl%
\url{https://openreview.net/forum?id=fAAaT826Vv}
\showURL{%
\tempurl}


\bibitem[Feo et~al\mbox{.}(1990)]%
        {sisal}
\bibfield{author}{\bibinfo{person}{John~T. Feo}, \bibinfo{person}{David~C. Cann}, {and} \bibinfo{person}{Rodney~R. Oldehoeft}.} \bibinfo{year}{1990}\natexlab{}.
\newblock \showarticletitle{A report on the sisal language project}.
\newblock \bibinfo{journal}{\emph{J. Parallel and Distrib. Comput.}} \bibinfo{volume}{10}, \bibinfo{number}{4} (\bibinfo{year}{1990}), \bibinfo{pages}{349--366}.
\newblock
\shownote{Data-flow Processing}.
\newblock
\showISSN{0743-7315}
\href{https://doi.org/10.1016/0743-7315(90)90035-N}{doi:\nolinkurl{10.1016/0743-7315(90)90035-N}}


\bibitem[Flanagan et~al\mbox{.}(1993)]%
        {anf}
\bibfield{author}{\bibinfo{person}{Cormac Flanagan}, \bibinfo{person}{Amr Sabry}, \bibinfo{person}{Bruce~F. Duba}, {and} \bibinfo{person}{Matthias Felleisen}.} \bibinfo{year}{1993}\natexlab{}.
\newblock \showarticletitle{The essence of compiling with continuations}. In \bibinfo{booktitle}{\emph{Proceedings of the ACM SIGPLAN 1993 Conference on Programming Language Design and Implementation}} (Albuquerque, New Mexico, USA) \emph{(\bibinfo{series}{PLDI '93})}. \bibinfo{publisher}{Association for Computing Machinery}, \bibinfo{address}{New York, NY, USA}, \bibinfo{pages}{237–247}.
\newblock
\showISBNx{0897915984}
\href{https://doi.org/10.1145/155090.155113}{doi:\nolinkurl{10.1145/155090.155113}}


\bibitem[Fraser et~al\mbox{.}(1984)]%
        {fraser1984analyzing}
\bibfield{author}{\bibinfo{person}{Christopher~W. Fraser}, \bibinfo{person}{Eugene~W. Myers}, {and} \bibinfo{person}{Alan~L. Wendt}.} \bibinfo{year}{1984}\natexlab{}.
\newblock \showarticletitle{Analyzing and Compressing Assembly Code}. In \bibinfo{booktitle}{\emph{Proceedings of the 1984 SIGPLAN Symposium on Compiler Construction}}. \bibinfo{publisher}{ACM}, \bibinfo{pages}{117--121}.
\newblock
\href{https://doi.org/10.1145/502874.502886}{doi:\nolinkurl{10.1145/502874.502886}}


\bibitem[{Gemini Team, Google}(2023)]%
        {gemini}
\bibfield{author}{\bibinfo{person}{{Gemini Team, Google}}.} \bibinfo{year}{2023}\natexlab{}.
\newblock \showarticletitle{Gemini: A Family of Highly Capable Multimodal Models}.
\newblock \bibinfo{journal}{\emph{arXiv preprint arXiv:2312.11805}} (\bibinfo{year}{2023}).
\newblock


\bibitem[{Google}(2025)]%
        {adk}
\bibfield{author}{\bibinfo{person}{{Google}}.} \bibinfo{year}{2025}\natexlab{}.
\newblock \bibinfo{booktitle}{\emph{{Agent Development Kit (ADK)}: An Open-Source, Code-First Python Toolkit for Building, Evaluating, and Deploying Sophisticated AI Agents}}.
\newblock
\urldef\tempurl%
\url{https://github.com/google/adk-python}
\showURL{%
\tempurl}


\bibitem[Halstead(1985)]%
        {multilisp}
\bibfield{author}{\bibinfo{person}{Robert~H. Halstead}.} \bibinfo{year}{1985}\natexlab{}.
\newblock \showarticletitle{MULTILISP: a language for concurrent symbolic computation}.
\newblock \bibinfo{journal}{\emph{ACM Trans. Program. Lang. Syst.}} \bibinfo{volume}{7}, \bibinfo{number}{4} (\bibinfo{date}{Oct.} \bibinfo{year}{1985}), \bibinfo{pages}{501–538}.
\newblock
\showISSN{0164-0925}
\href{https://doi.org/10.1145/4472.4478}{doi:\nolinkurl{10.1145/4472.4478}}


\bibitem[He and Roy(2025)]%
        {he2025logictree}
\bibfield{author}{\bibinfo{person}{Kang He} {and} \bibinfo{person}{Kaushik Roy}.} \bibinfo{year}{2025}\natexlab{}.
\newblock \showarticletitle{LogicTree: Structured Proof Exploration for Coherent and Rigorous Logical Reasoning with Large Language Models}.
\newblock \bibinfo{journal}{\emph{arXiv preprint arXiv:2504.14089}} (\bibinfo{year}{2025}).
\newblock


\bibitem[Huang et~al\mbox{.}(2026)]%
        {concur-benchmark}
\bibfield{author}{\bibinfo{person}{Jue Huang}, \bibinfo{person}{Tarek Mahmud}, \bibinfo{person}{Corina Pasareanu}, {and} \bibinfo{person}{Guowei Yang}.} \bibinfo{year}{2026}\natexlab{}.
\newblock \bibinfo{title}{{CONCUR}: Benchmarking {LLM}s for Concurrent Code Generation}.
\newblock
\showeprint[arxiv]{2603.03683}~[cs.SE]
\urldef\tempurl%
\url{https://arxiv.org/abs/2603.03683}
\showURL{%
\tempurl}


\bibitem[Hudson and Manning(2019)]%
        {gqa}
\bibfield{author}{\bibinfo{person}{Drew~A. Hudson} {and} \bibinfo{person}{Christopher~D. Manning}.} \bibinfo{year}{2019}\natexlab{}.
\newblock \showarticletitle{{GQA}: A New Dataset for Real-World Visual Reasoning and Compositional Question Answering}. In \bibinfo{booktitle}{\emph{2019 IEEE/CVF Conference on Computer Vision and Pattern Recognition (CVPR)}}. \bibinfo{publisher}{IEEE}, \bibinfo{pages}{6693--6702}.
\newblock
\href{https://doi.org/10.1109/CVPR.2019.00686}{doi:\nolinkurl{10.1109/CVPR.2019.00686}}


\bibitem[Jain and Purandare(2025)]%
        {jain2025assessing}
\bibfield{author}{\bibinfo{person}{Ridhi Jain} {and} \bibinfo{person}{Rahul Purandare}.} \bibinfo{year}{2025}\natexlab{}.
\newblock \bibinfo{title}{Assessing Large Language Models in Comprehending and Verifying Concurrent Programs across Memory Models}.
\newblock
\showeprint[arxiv]{2501.14326}~[cs.SE]
\urldef\tempurl%
\url{https://arxiv.org/abs/2501.14326}
\showURL{%
\tempurl}


\bibitem[{JetBrains} and {Python Software Foundation}(2025)]%
        {jetbrains2024}
\bibfield{author}{\bibinfo{person}{{JetBrains}} {and} \bibinfo{person}{{Python Software Foundation}}.} \bibinfo{year}{2025}\natexlab{}.
\newblock \bibinfo{title}{Python Developers Survey 2024 Results}.
\newblock \bibinfo{howpublished}{\url{https://lp.jetbrains.com/python-developers-survey-2024/}}.
\newblock
\shownote{Accessed: 2026-09-09}.
\newblock


\bibitem[Jungmair et~al\mbox{.}(2024)]%
        {hipy}
\bibfield{author}{\bibinfo{person}{Michael Jungmair}, \bibinfo{person}{Alexis Engelke}, {and} \bibinfo{person}{Jana Giceva}.} \bibinfo{year}{2024}\natexlab{}.
\newblock \showarticletitle{HiPy: Extracting High-Level Semantics from Python Code for Data Processing}.
\newblock \bibinfo{journal}{\emph{Proc. ACM Program. Lang.}} \bibinfo{volume}{8}, \bibinfo{number}{OOPSLA2}, Article \bibinfo{articleno}{297} (\bibinfo{date}{Oct.} \bibinfo{year}{2024}), \bibinfo{numpages}{27}~pages.
\newblock
\href{https://doi.org/10.1145/3689737}{doi:\nolinkurl{10.1145/3689737}}


\bibitem[Kallas et~al\mbox{.}(2022)]%
        {pash-jit-osdi}
\bibfield{author}{\bibinfo{person}{Konstantinos Kallas}, \bibinfo{person}{Tammam Mustafa}, \bibinfo{person}{Jan Bielak}, \bibinfo{person}{Dimitris Karnikis}, \bibinfo{person}{Thurston~HY Dang}, \bibinfo{person}{Michael Greenberg}, {and} \bibinfo{person}{Nikos Vasilakis}.} \bibinfo{year}{2022}\natexlab{}.
\newblock \showarticletitle{Practically correct,{Just-in-Time} shell script parallelization}. In \bibinfo{booktitle}{\emph{16th USENIX Symposium on Operating Systems Design and Implementation (OSDI 22)}}. \bibinfo{pages}{769--785}.
\newblock


\bibitem[Karp and Miller(1966)]%
        {Karp1966}
\bibfield{author}{\bibinfo{person}{Richard~M. Karp} {and} \bibinfo{person}{Raymond Miller}.} \bibinfo{year}{1966}\natexlab{}.
\newblock \showarticletitle{{Properties of a Model for Parallel Computation: Determinacy, Termination, Queueing}}.
\newblock \bibinfo{journal}{\emph{SlAM J. of Applied Mathematics}} \bibinfo{volume}{14}, \bibinfo{number}{6} (\bibinfo{date}{11} \bibinfo{year}{1966}), \bibinfo{pages}{1390--1411}.
\newblock


\bibitem[Karp and Miller(1969)]%
        {Karp1969}
\bibfield{author}{\bibinfo{person}{Richard~M. Karp} {and} \bibinfo{person}{Raymond Miller}.} \bibinfo{year}{1969}\natexlab{}.
\newblock \showarticletitle{Parallel Program Schemata}.
\newblock \bibinfo{journal}{\emph{J. Comput. System Sci.}}  \bibinfo{volume}{3} (\bibinfo{year}{1969}), \bibinfo{pages}{147--195}.
\newblock


\bibitem[Kelsey(1995)]%
        {kelsey1995correspondence}
\bibfield{author}{\bibinfo{person}{Richard~A. Kelsey}.} \bibinfo{year}{1995}\natexlab{}.
\newblock \showarticletitle{A Correspondence between Continuation Passing Style and Static Single Assignment Form}. In \bibinfo{booktitle}{\emph{Papers from the 1995 ACM SIGPLAN Workshop on Intermediate Representations}} \emph{(\bibinfo{series}{IR '95})}. \bibinfo{publisher}{ACM}, \bibinfo{pages}{13--22}.
\newblock
\href{https://doi.org/10.1145/202529.202532}{doi:\nolinkurl{10.1145/202529.202532}}


\bibitem[Lam et~al\mbox{.}(2015)]%
        {numba}
\bibfield{author}{\bibinfo{person}{Siu~Kwan Lam}, \bibinfo{person}{Antoine Pitrou}, {and} \bibinfo{person}{Stanley Seibert}.} \bibinfo{year}{2015}\natexlab{}.
\newblock \showarticletitle{Numba: A llvm-based python jit compiler}. In \bibinfo{booktitle}{\emph{Proceedings of the Second Workshop on the LLVM Compiler Infrastructure in HPC}}. \bibinfo{pages}{1--6}.
\newblock


\bibitem[{LangChain Inc.}(2024)]%
        {langgraph}
\bibfield{author}{\bibinfo{person}{{LangChain Inc.}}} \bibinfo{year}{2024}\natexlab{}.
\newblock \bibinfo{booktitle}{\emph{{LangGraph}: Build Resilient Language Agents as Graphs}}.
\newblock
\urldef\tempurl%
\url{https://github.com/langchain-ai/langgraph}
\showURL{%
\tempurl}


\bibitem[Lewis et~al\mbox{.}(2020)]%
        {rag}
\bibfield{author}{\bibinfo{person}{Patrick Lewis}, \bibinfo{person}{Ethan Perez}, \bibinfo{person}{Aleksandra Piktus}, \bibinfo{person}{Fabio Petroni}, \bibinfo{person}{Vladimir Karpukhin}, \bibinfo{person}{Naman Goyal}, \bibinfo{person}{Heinrich K{\"u}ttler}, \bibinfo{person}{Mike Lewis}, \bibinfo{person}{Wen-tau Yih}, \bibinfo{person}{Tim Rockt{\"a}schel}, {et~al\mbox{.}}} \bibinfo{year}{2020}\natexlab{}.
\newblock \showarticletitle{Retrieval-augmented generation for knowledge-intensive nlp tasks}.
\newblock \bibinfo{journal}{\emph{Advances in neural information processing systems}}  \bibinfo{volume}{33} (\bibinfo{year}{2020}), \bibinfo{pages}{9459--9474}.
\newblock


\bibitem[Li et~al\mbox{.}(2024)]%
        {traq}
\bibfield{author}{\bibinfo{person}{Shuo Li}, \bibinfo{person}{Sangdon Park}, \bibinfo{person}{Insup Lee}, {and} \bibinfo{person}{Osbert Bastani}.} \bibinfo{year}{2024}\natexlab{}.
\newblock \showarticletitle{{TRAQ}: Trustworthy Retrieval Augmented Question Answering via Conformal Prediction}. In \bibinfo{booktitle}{\emph{Proceedings of the 2024 Conference of the North American Chapter of the Association for Computational Linguistics: Human Language Technologies (Volume 1: Long Papers)}}, \bibfield{editor}{\bibinfo{person}{Kevin Duh}, \bibinfo{person}{Helena Gomez}, {and} \bibinfo{person}{Steven Bethard}} (Eds.). \bibinfo{publisher}{Association for Computational Linguistics}, \bibinfo{address}{Mexico City, Mexico}, \bibinfo{pages}{3799--3821}.
\newblock
\href{https://doi.org/10.18653/v1/2024.naacl-long.210}{doi:\nolinkurl{10.18653/v1/2024.naacl-long.210}}


\bibitem[Liu(2022)]%
        {llamaindex}
\bibfield{author}{\bibinfo{person}{Jerry Liu}.} \bibinfo{year}{2022}\natexlab{}.
\newblock \bibinfo{booktitle}{\emph{{LlamaIndex}: A Data Framework for Building Agentic LLM Applications}}.
\newblock
\urldef\tempurl%
\url{https://github.com/run-llama/llama_index}
\showURL{%
\tempurl}


\bibitem[Lu et~al\mbox{.}(2008)]%
        {lu2008learning}
\bibfield{author}{\bibinfo{person}{Shan Lu}, \bibinfo{person}{Soyeon Park}, \bibinfo{person}{Eunsoo Seo}, {and} \bibinfo{person}{Yuanyuan Zhou}.} \bibinfo{year}{2008}\natexlab{}.
\newblock \showarticletitle{Learning from Mistakes: A Comprehensive Study on Real World Concurrency Bug Characteristics}. In \bibinfo{booktitle}{\emph{Proceedings of the 13th International Conference on Architectural Support for Programming Languages and Operating Systems (ASPLOS '08)}}. \bibinfo{publisher}{ACM}, \bibinfo{pages}{329--339}.
\newblock
\href{https://doi.org/10.1145/1346281.1346323}{doi:\nolinkurl{10.1145/1346281.1346323}}


\bibitem[Maessen et~al\mbox{.}(1995)]%
        {maessen1995semantics}
\bibfield{author}{\bibinfo{person}{Shail Aditya Arvind Jan-Willem Maessen}, \bibinfo{person}{Lennart Augustsson}, {and} \bibinfo{person}{Rishiyur~S Nikhil}.} \bibinfo{year}{1995}\natexlab{}.
\newblock \showarticletitle{Semantics of pH: A parallel dialect of Haskell}. In \bibinfo{booktitle}{\emph{In Proceedings from the Haskell Workshop (at FPCA 95)}}. \bibinfo{pages}{35--49}.
\newblock


\bibitem[McGraw(1982)]%
        {mcgraw1982val}
\bibfield{author}{\bibinfo{person}{James~R McGraw}.} \bibinfo{year}{1982}\natexlab{}.
\newblock \showarticletitle{The VAL language: Description and analysis}.
\newblock \bibinfo{journal}{\emph{ACM Transactions on Programming Languages and Systems (TOPLAS)}} \bibinfo{volume}{4}, \bibinfo{number}{1} (\bibinfo{year}{1982}), \bibinfo{pages}{44--82}.
\newblock


\bibitem[Mell et~al\mbox{.}(2025)]%
        {opal}
\bibfield{author}{\bibinfo{person}{Stephen Mell}, \bibinfo{person}{Konstantinos Kallas}, \bibinfo{person}{Steve Zdancewic}, {and} \bibinfo{person}{Osbert Bastani}.} \bibinfo{year}{2025}\natexlab{}.
\newblock \showarticletitle{Opportunistically Parallel Lambda Calculus}.
\newblock \bibinfo{journal}{\emph{Proc. ACM Program. Lang.}} \bibinfo{volume}{9}, \bibinfo{number}{OOPSLA2}, Article \bibinfo{articleno}{365} (\bibinfo{date}{Oct.} \bibinfo{year}{2025}), \bibinfo{numpages}{27}~pages.
\newblock
\href{https://doi.org/10.1145/3763143}{doi:\nolinkurl{10.1145/3763143}}


\bibitem[{Microsoft}(2025)]%
        {maf}
\bibfield{author}{\bibinfo{person}{{Microsoft}}.} \bibinfo{year}{2025}\natexlab{}.
\newblock \bibinfo{booktitle}{\emph{{Microsoft Agent Framework}: A Framework for Building, Orchestrating and Deploying AI Agents and Multi-Agent Workflows}}.
\newblock
\urldef\tempurl%
\url{https://github.com/microsoft/agent-framework}
\showURL{%
\tempurl}


\bibitem[Mustafa et~al\mbox{.}(2023)]%
        {dish}
\bibfield{author}{\bibinfo{person}{Tammam Mustafa}, \bibinfo{person}{Konstantinos Kallas}, \bibinfo{person}{Pratyush Das}, {and} \bibinfo{person}{Nikos Vasilakis}.} \bibinfo{year}{2023}\natexlab{}.
\newblock \showarticletitle{{DiSh}: Dynamic {Shell-Script} Distribution}. In \bibinfo{booktitle}{\emph{20th USENIX Symposium on Networked Systems Design and Implementation (NSDI 23)}}. \bibinfo{pages}{341--356}.
\newblock


\bibitem[Ning et~al\mbox{.}(2024)]%
        {skeletonofthought}
\bibfield{author}{\bibinfo{person}{Xuefei Ning}, \bibinfo{person}{Zinan Lin}, \bibinfo{person}{Zixuan Zhou}, \bibinfo{person}{Zifu Wang}, \bibinfo{person}{Huazhong Yang}, {and} \bibinfo{person}{Yu Wang}.} \bibinfo{year}{2024}\natexlab{}.
\newblock \showarticletitle{Skeleton-of-Thought: Prompting {LLM}s for Efficient Parallel Generation}. In \bibinfo{booktitle}{\emph{The Twelfth International Conference on Learning Representations}}.
\newblock
\urldef\tempurl%
\url{https://openreview.net/forum?id=mqVgBbNCm9}
\showURL{%
\tempurl}


\bibitem[{OpenAI}(2023)]%
        {gpt}
\bibfield{author}{\bibinfo{person}{{OpenAI}}.} \bibinfo{year}{2023}\natexlab{}.
\newblock \bibinfo{booktitle}{\emph{{GPT}-4 Technical Report}}.
\newblock \bibinfo{type}{{T}echnical {R}eport}. \bibinfo{institution}{OpenAI}.
\newblock
\shownote{arXiv preprint arXiv:2303.08774}.
\newblock


\bibitem[Palkar et~al\mbox{.}(2017)]%
        {weld}
\bibfield{author}{\bibinfo{person}{Shoumik Palkar}, \bibinfo{person}{James~J Thomas}, \bibinfo{person}{Anil Shanbhag}, \bibinfo{person}{Deepak Narayanan}, \bibinfo{person}{Holger Pirk}, \bibinfo{person}{Malte Schwarzkopf}, \bibinfo{person}{Saman Amarasinghe}, {and} \bibinfo{person}{Matei Zaharia}.} \bibinfo{year}{2017}\natexlab{}.
\newblock \showarticletitle{Weld: A common runtime for high performance data analytics}.
\newblock  (\bibinfo{year}{2017}).
\newblock


\bibitem[Palkar and Zaharia(2019)]%
        {mozart}
\bibfield{author}{\bibinfo{person}{Shoumik Palkar} {and} \bibinfo{person}{Matei Zaharia}.} \bibinfo{year}{2019}\natexlab{}.
\newblock \showarticletitle{Optimizing Data-Intensive Computations in Existing Libraries with Split Annotations}. In \bibinfo{booktitle}{\emph{Proceedings of the 27th ACM Symposium on Operating Systems Principles}} (Huntsville, ON, Canada) \emph{(\bibinfo{series}{SOSP '19})}. \bibinfo{publisher}{ACM}, \bibinfo{pages}{291--305}.
\newblock
\href{https://doi.org/10.1145/3341301.3359652}{doi:\nolinkurl{10.1145/3341301.3359652}}


\bibitem[Pietsch et~al\mbox{.}(2019)]%
        {haystack}
\bibfield{author}{\bibinfo{person}{Malte Pietsch}, \bibinfo{person}{Timo M{\"o}ller}, \bibinfo{person}{Bogdan Kostic}, \bibinfo{person}{Julian Risch}, \bibinfo{person}{Massimiliano Pippi}, \bibinfo{person}{Mayank Jobanputra}, \bibinfo{person}{Sara Zanzottera}, \bibinfo{person}{Silvano Cerza}, \bibinfo{person}{Vladimir Blagojevic}, \bibinfo{person}{Thomas Stadelmann}, \bibinfo{person}{Tanay Soni}, {and} \bibinfo{person}{Sebastian Lee}.} \bibinfo{year}{2019}\natexlab{}.
\newblock \bibinfo{booktitle}{\emph{{Haystack}: The End-to-End {NLP} Framework for Pragmatic Builders}}.
\newblock
\urldef\tempurl%
\url{https://github.com/deepset-ai/haystack}
\showURL{%
\tempurl}


\bibitem[Politz et~al\mbox{.}(2013b)]%
        {lambdapi}
\bibfield{author}{\bibinfo{person}{Joe~Gibbs Politz}, \bibinfo{person}{Alejandro Martinez}, \bibinfo{person}{Mae Milano}, \bibinfo{person}{Sumner Warren}, \bibinfo{person}{Daniel Patterson}, \bibinfo{person}{Junsong Li}, \bibinfo{person}{Anand Chitipothu}, {and} \bibinfo{person}{Shriram Krishnamurthi}.} \bibinfo{year}{2013}\natexlab{b}.
\newblock \showarticletitle{Python: the full monty}.
\newblock \bibinfo{journal}{\emph{SIGPLAN Not.}} \bibinfo{volume}{48}, \bibinfo{number}{10} (\bibinfo{date}{Oct.} \bibinfo{year}{2013}), \bibinfo{pages}{217–232}.
\newblock
\showISSN{0362-1340}
\href{https://doi.org/10.1145/2544173.2509536}{doi:\nolinkurl{10.1145/2544173.2509536}}


\bibitem[Politz et~al\mbox{.}(2013a)]%
        {politz2013python}
\bibfield{author}{\bibinfo{person}{Joe~Gibbs Politz}, \bibinfo{person}{Alejandro Martinez}, \bibinfo{person}{Matthew Milano}, \bibinfo{person}{Sumner Warren}, \bibinfo{person}{Daniel Patterson}, \bibinfo{person}{Junsong Li}, {and} \bibinfo{person}{Shriram Krishnamurthi}.} \bibinfo{year}{2013}\natexlab{a}.
\newblock \showarticletitle{Python: The Full Monty — A Tested Semantics for the Python Programming Language}. In \bibinfo{booktitle}{\emph{OOPSLA}}.
\newblock


\bibitem[{Python Software Foundation}(2026a)]%
        {python-asyncio}
\bibfield{author}{\bibinfo{person}{{Python Software Foundation}}.} \bibinfo{year}{2026}\natexlab{a}.
\newblock \bibinfo{title}{{asyncio}: Asynchronous {I/O}}.
\newblock \bibinfo{howpublished}{\url{https://docs.python.org/3/library/asyncio.html}}.
\newblock
\shownote{Python documentation. Accessed September 10, 2026}.
\newblock


\bibitem[{Python Software Foundation}(2026b)]%
        {python-threading}
\bibfield{author}{\bibinfo{person}{{Python Software Foundation}}.} \bibinfo{year}{2026}\natexlab{b}.
\newblock \bibinfo{title}{{threading}: Thread-based Parallelism}.
\newblock \bibinfo{howpublished}{\url{https://docs.python.org/3/library/threading.html}}.
\newblock
\shownote{Python documentation. Accessed September 10, 2026}.
\newblock


\bibitem[Raghavan et~al\mbox{.}(2020)]%
        {posh}
\bibfield{author}{\bibinfo{person}{Deepti Raghavan}, \bibinfo{person}{Sadjad Fouladi}, \bibinfo{person}{Philip Levis}, {and} \bibinfo{person}{Matei Zaharia}.} \bibinfo{year}{2020}\natexlab{}.
\newblock \showarticletitle{{POSH}: A {Data-Aware} Shell}. In \bibinfo{booktitle}{\emph{2020 USENIX Annual Technical Conference (USENIX ATC 20)}}. \bibinfo{pages}{617--631}.
\newblock


\bibitem[Rodriguez(1969)]%
        {Rodriguez1969}
\bibfield{author}{\bibinfo{person}{Jorge~E. Rodriguez}.} \bibinfo{year}{1969}\natexlab{}.
\newblock \emph{\bibinfo{title}{A Graph Model for Parallel Computations}}.
\newblock \bibinfo{thesistype}{Ph.\,D. Dissertation}.
\newblock
\shownote{MIT-LCS-TR64}.
\newblock


\bibitem[Romera-Paredes et~al\mbox{.}(2024)]%
        {funsearch}
\bibfield{author}{\bibinfo{person}{Bernardino Romera-Paredes}, \bibinfo{person}{Mohammadamin Barekatain}, \bibinfo{person}{Alexander Novikov}, \bibinfo{person}{Matej Balog}, \bibinfo{person}{M.~Pawan Kumar}, \bibinfo{person}{Emilien Dupont}, \bibinfo{person}{Francisco J.~R. Ruiz}, \bibinfo{person}{Jordan~S. Ellenberg}, \bibinfo{person}{Pengming Wang}, \bibinfo{person}{Omar Fawzi}, \bibinfo{person}{Pushmeet Kohli}, {and} \bibinfo{person}{Alhussein Fawzi}.} \bibinfo{year}{2024}\natexlab{}.
\newblock \showarticletitle{Mathematical discoveries from program search with large language models}.
\newblock \bibinfo{journal}{\emph{Nature}} \bibinfo{volume}{625}, \bibinfo{number}{7995} (\bibinfo{year}{2024}), \bibinfo{pages}{468--475}.
\newblock
\href{https://doi.org/10.1038/s41586-023-06924-6}{doi:\nolinkurl{10.1038/s41586-023-06924-6}}


\bibitem[Sastry and Ju(1998)]%
        {register-promotion}
\bibfield{author}{\bibinfo{person}{A.V.S. Sastry} {and} \bibinfo{person}{Roy~D.C. Ju}.} \bibinfo{year}{1998}\natexlab{}.
\newblock \showarticletitle{A New Algorithm for Scalar Register Promotion Based on SSA Form}.
\newblock \bibinfo{journal}{\emph{PLDI '98: Proceedings of the ACM SIGPLAN 1998 conference on Programming language design and implementation}} (\bibinfo{year}{1998}).
\newblock


\bibitem[Shajii et~al\mbox{.}(2023)]%
        {codon}
\bibfield{author}{\bibinfo{person}{Ariya Shajii}, \bibinfo{person}{Gabriel Ramirez}, \bibinfo{person}{Haris Smajlovi\'{c}}, \bibinfo{person}{Jessica Ray}, \bibinfo{person}{Bonnie Berger}, \bibinfo{person}{Saman Amarasinghe}, {and} \bibinfo{person}{Ibrahim Numanagi\'{c}}.} \bibinfo{year}{2023}\natexlab{}.
\newblock \showarticletitle{Codon: A Compiler for High-Performance Pythonic Applications and DSLs}. In \bibinfo{booktitle}{\emph{Proceedings of the 32nd ACM SIGPLAN International Conference on Compiler Construction}} (Montr\'{e}al, QC, Canada) \emph{(\bibinfo{series}{CC 2023})}. \bibinfo{publisher}{Association for Computing Machinery}, \bibinfo{address}{New York, NY, USA}, \bibinfo{pages}{191–202}.
\newblock
\showISBNx{9798400700880}
\href{https://doi.org/10.1145/3578360.3580275}{doi:\nolinkurl{10.1145/3578360.3580275}}


\bibitem[Shao et~al\mbox{.}(2026)]%
        {shao2026linguistic}
\bibfield{author}{\bibinfo{person}{Shuai Shao}, \bibinfo{person}{Lu Xiao}, {and} \bibinfo{person}{Tingting Yu}.} \bibinfo{year}{2026}\natexlab{}.
\newblock \bibinfo{title}{Identifying Concurrency Bug Reports via Linguistic Patterns}.
\newblock
\showeprint[arxiv]{2601.16338}~[cs.SE]
\urldef\tempurl%
\url{https://arxiv.org/abs/2601.16338}
\showURL{%
\tempurl}


\bibitem[Silva et~al\mbox{.}(2024)]%
        {totdomainmodeling}
\bibfield{author}{\bibinfo{person}{Jonathan Silva}, \bibinfo{person}{Qin Ma}, \bibinfo{person}{Jordi Cabot}, \bibinfo{person}{Pierre Kelsen}, {and} \bibinfo{person}{Henderik~A. Proper}.} \bibinfo{year}{2024}\natexlab{}.
\newblock \showarticletitle{Application of the Tree-of-Thoughts Framework to LLM-Enabled Domain Modeling}. In \bibinfo{booktitle}{\emph{Conceptual Modeling: 43rd International Conference, ER 2024, Pittsburgh, PA, USA, October 28–31, 2024, Proceedings}} (Pittsburg, PA, USA). \bibinfo{publisher}{Springer-Verlag}, \bibinfo{address}{Berlin, Heidelberg}, \bibinfo{pages}{94–111}.
\newblock
\showISBNx{978-3-031-75871-3}
\href{https://doi.org/10.1007/978-3-031-75872-0_6}{doi:\nolinkurl{10.1007/978-3-031-75872-0_6}}


\bibitem[Spiegelberg et~al\mbox{.}(2021)]%
        {tuplex}
\bibfield{author}{\bibinfo{person}{Leonhard Spiegelberg}, \bibinfo{person}{Rahul Yesantharao}, \bibinfo{person}{Malte Schwarzkopf}, {and} \bibinfo{person}{Tim Kraska}.} \bibinfo{year}{2021}\natexlab{}.
\newblock \showarticletitle{Tuplex: Data Science in Python at Native Code Speed}. In \bibinfo{booktitle}{\emph{Proceedings of the 2021 International Conference on Management of Data}} (Virtual Event, China) \emph{(\bibinfo{series}{SIGMOD '21})}. \bibinfo{publisher}{Association for Computing Machinery}, \bibinfo{address}{New York, NY, USA}, \bibinfo{pages}{1718–1731}.
\newblock
\showISBNx{9781450383431}
\href{https://doi.org/10.1145/3448016.3457244}{doi:\nolinkurl{10.1145/3448016.3457244}}


\bibitem[Suris et~al\mbox{.}(2023)]%
        {suris2023vipergpt}
\bibfield{author}{\bibinfo{person}{David Suris} {et~al\mbox{.}}} \bibinfo{year}{2023}\natexlab{}.
\newblock \showarticletitle{ViperGPT: Visual Inference via Python Execution for Reasoning}.
\newblock \bibinfo{journal}{\emph{arXiv preprint arXiv:2303.08128}} (\bibinfo{year}{2023}).
\newblock


\bibitem[{The Apache Software Foundation}(2024)]%
        {hamilton}
\bibfield{author}{\bibinfo{person}{{The Apache Software Foundation}}.} \bibinfo{year}{2024}\natexlab{}.
\newblock \bibinfo{booktitle}{\emph{{Apache Hamilton} (incubating): Define Testable, Modular, Self-Documenting Dataflows}}.
\newblock
\urldef\tempurl%
\url{https://github.com/apache/hamilton}
\showURL{%
\tempurl}


\bibitem[Trinh et~al\mbox{.}(2024)]%
        {alphageometry}
\bibfield{author}{\bibinfo{person}{Trieu Trinh}, \bibinfo{person}{Yuhuai Wu}, \bibinfo{person}{Quoc~V. Le}, \bibinfo{person}{He He}, {and} \bibinfo{person}{Minh-Thang Luong}.} \bibinfo{year}{2024}\natexlab{}.
\newblock \showarticletitle{Solving Olympiad Geometry without Human Demonstrations}.
\newblock \bibinfo{journal}{\emph{Nature}}  \bibinfo{volume}{625} (\bibinfo{year}{2024}), \bibinfo{pages}{476--482}.
\newblock
\href{https://doi.org/10.1038/s41586-023-06747-5}{doi:\nolinkurl{10.1038/s41586-023-06747-5}}


\bibitem[Vasilakis et~al\mbox{.}(2021)]%
        {pash:eurosys:2021}
\bibfield{author}{\bibinfo{person}{Nikos Vasilakis}, \bibinfo{person}{Konstantinos Kallas}, \bibinfo{person}{Konstantinos Mamouras}, \bibinfo{person}{Achilles Benetopoulos}, {and} \bibinfo{person}{Lazar Cvetkovi\'{c}}.} \bibinfo{year}{2021}\natexlab{}.
\newblock \showarticletitle{PaSh: Light-Touch Data-Parallel Shell Processing}. In \bibinfo{booktitle}{\emph{Proceedings of the Sixteenth European Conference on Computer Systems}} (Online Event, United Kingdom) \emph{(\bibinfo{series}{EuroSys '21})}. \bibinfo{publisher}{Association for Computing Machinery}, \bibinfo{address}{New York, NY, USA}, \bibinfo{pages}{49--66}.
\newblock
\showISBNx{9781450383349}
\href{https://doi.org/10.1145/3447786.3456228}{doi:\nolinkurl{10.1145/3447786.3456228}}


\bibitem[Vasilakis et~al\mbox{.}(2019)]%
        {ignis:pldi:2019}
\bibfield{author}{\bibinfo{person}{Nikos Vasilakis}, \bibinfo{person}{Ben Karel}, \bibinfo{person}{Yash Palkhiwala}, \bibinfo{person}{John Sonchack}, \bibinfo{person}{Andr{\'e} DeHon}, {and} \bibinfo{person}{Jonathan~M. Smith}.} \bibinfo{year}{2019}\natexlab{}.
\newblock \showarticletitle{Ignis: Scaling Distribution-Oblivious Systems with Light-Touch Distribution}. In \bibinfo{booktitle}{\emph{Proceedings of the 40th ACM SIGPLAN Conference on Programming Language Design and Implementation}} (Phoenix, AZ, USA) \emph{(\bibinfo{series}{PLDI 2019})}. \bibinfo{publisher}{ACM}, \bibinfo{pages}{1010--1026}.
\newblock
\href{https://doi.org/10.1145/3314221.3314586}{doi:\nolinkurl{10.1145/3314221.3314586}}


\bibitem[Wadler(1990)]%
        {wadler-monads}
\bibfield{author}{\bibinfo{person}{Philip Wadler}.} \bibinfo{year}{1990}\natexlab{}.
\newblock \showarticletitle{Comprehending monads}. In \bibinfo{booktitle}{\emph{Proceedings of the 1990 ACM Conference on LISP and Functional Programming}} (Nice, France) \emph{(\bibinfo{series}{LFP '90})}. \bibinfo{publisher}{Association for Computing Machinery}, \bibinfo{address}{New York, NY, USA}, \bibinfo{pages}{61–78}.
\newblock
\showISBNx{089791368X}
\href{https://doi.org/10.1145/91556.91592}{doi:\nolinkurl{10.1145/91556.91592}}


\bibitem[Wang et~al\mbox{.}(2024)]%
        {codeact}
\bibfield{author}{\bibinfo{person}{Xingyao Wang}, \bibinfo{person}{Yangyi Chen}, \bibinfo{person}{Lifan Yuan}, \bibinfo{person}{Yizhe Zhang}, \bibinfo{person}{Yunzhu Li}, \bibinfo{person}{Hao Peng}, {and} \bibinfo{person}{Heng Ji}.} \bibinfo{year}{2024}\natexlab{}.
\newblock \showarticletitle{Executable Code Actions Elicit Better {LLM} Agents}. In \bibinfo{booktitle}{\emph{Proceedings of the 41st International Conference on Machine Learning}} \emph{(\bibinfo{series}{Proceedings of Machine Learning Research}, Vol.~\bibinfo{volume}{235})}, \bibfield{editor}{\bibinfo{person}{Ruslan Salakhutdinov}, \bibinfo{person}{Zico Kolter}, \bibinfo{person}{Katherine Heller}, \bibinfo{person}{Adrian Weller}, \bibinfo{person}{Nuria Oliver}, \bibinfo{person}{Jonathan Scarlett}, {and} \bibinfo{person}{Felix Berkenkamp}} (Eds.). \bibinfo{publisher}{PMLR}, \bibinfo{pages}{50208--50232}.
\newblock
\urldef\tempurl%
\url{https://proceedings.mlr.press/v235/wang24h.html}
\showURL{%
\tempurl}


\bibitem[Whiting and Pascoe(1994)]%
        {dataflow-history}
\bibfield{author}{\bibinfo{person}{Paul~G. Whiting} {and} \bibinfo{person}{Robert S.~V. Pascoe}.} \bibinfo{year}{1994}\natexlab{}.
\newblock \showarticletitle{A history of data-flow languages}.
\newblock \bibinfo{journal}{\emph{IEEE Annals of the History of Computing}}  \bibinfo{volume}{16} (\bibinfo{year}{1994}), \bibinfo{pages}{38--59}.
\newblock
\urldef\tempurl%
\url{https://api.semanticscholar.org/CorpusID:7384421}
\showURL{%
\tempurl}


\bibitem[Xu et~al\mbox{.}(2025)]%
        {xu2025embodiedtreethoughtsdeliberate}
\bibfield{author}{\bibinfo{person}{Wenjiang Xu}, \bibinfo{person}{Cindy Wang}, \bibinfo{person}{Rui Fang}, \bibinfo{person}{Mingkang Zhang}, \bibinfo{person}{Lusong Li}, \bibinfo{person}{Jing Xu}, \bibinfo{person}{Jiayuan Gu}, \bibinfo{person}{Zecui Zeng}, {and} \bibinfo{person}{Rui Chen}.} \bibinfo{year}{2025}\natexlab{}.
\newblock \bibinfo{title}{Embodied Tree of Thoughts: Deliberate Manipulation Planning with Embodied World Model}.
\newblock
\showeprint[arxiv]{2512.08188}~[cs.RO]
\urldef\tempurl%
\url{https://arxiv.org/abs/2512.08188}
\showURL{%
\tempurl}


\bibitem[Yang et~al\mbox{.}(2024)]%
        {sweagent}
\bibfield{author}{\bibinfo{person}{John Yang}, \bibinfo{person}{Carlos~E Jimenez}, \bibinfo{person}{Alexander Wettig}, \bibinfo{person}{Kilian Lieret}, \bibinfo{person}{Shunyu Yao}, \bibinfo{person}{Karthik Narasimhan}, {and} \bibinfo{person}{Ofir Press}.} \bibinfo{year}{2024}\natexlab{}.
\newblock \showarticletitle{Swe-agent: Agent-computer interfaces enable automated software engineering}.
\newblock \bibinfo{journal}{\emph{Advances in Neural Information Processing Systems}}  \bibinfo{volume}{37} (\bibinfo{year}{2024}), \bibinfo{pages}{50528--50652}.
\newblock


\bibitem[Yang et~al\mbox{.}(2022)]%
        {yang2022complex-python}
\bibfield{author}{\bibinfo{person}{Yi Yang}, \bibinfo{person}{Ana Milanova}, {and} \bibinfo{person}{Martin Hirzel}.} \bibinfo{year}{2022}\natexlab{}.
\newblock \showarticletitle{Complex Python features in the wild}. In \bibinfo{booktitle}{\emph{Proceedings of the 19th International Conference on Mining Software Repositories}} (Pittsburgh, Pennsylvania) \emph{(\bibinfo{series}{MSR '22})}. \bibinfo{publisher}{Association for Computing Machinery}, \bibinfo{address}{New York, NY, USA}, \bibinfo{pages}{282–293}.
\newblock
\showISBNx{9781450393034}
\href{https://doi.org/10.1145/3524842.3528467}{doi:\nolinkurl{10.1145/3524842.3528467}}


\bibitem[Yao et~al\mbox{.}(2023)]%
        {treeofthoughts}
\bibfield{author}{\bibinfo{person}{Shunyu Yao}, \bibinfo{person}{Dian Yu}, \bibinfo{person}{Jeffrey Zhao}, \bibinfo{person}{Izhak Shafran}, \bibinfo{person}{Tom Griffiths}, \bibinfo{person}{Yuan Cao}, {and} \bibinfo{person}{Karthik Narasimhan}.} \bibinfo{year}{2023}\natexlab{}.
\newblock \showarticletitle{Tree of thoughts: Deliberate problem solving with large language models}.
\newblock \bibinfo{journal}{\emph{Advances in neural information processing systems}}  \bibinfo{volume}{36} (\bibinfo{year}{2023}), \bibinfo{pages}{11809--11822}.
\newblock


\bibitem[Yao et~al\mbox{.}(2022)]%
        {yao2022react}
\bibfield{author}{\bibinfo{person}{Shunyu Yao}, \bibinfo{person}{Jeffrey Zhao}, \bibinfo{person}{Dian Yu}, \bibinfo{person}{Nan Du}, \bibinfo{person}{Izhak Shafran}, \bibinfo{person}{Karthik~R Narasimhan}, {and} \bibinfo{person}{Yuan Cao}.} \bibinfo{year}{2022}\natexlab{}.
\newblock \showarticletitle{React: Synergizing reasoning and acting in language models}. In \bibinfo{booktitle}{\emph{The eleventh international conference on learning representations}}.
\newblock


\bibitem[Zhou et~al\mbox{.}(2024)]%
        {appy}
\bibfield{author}{\bibinfo{person}{Tong Zhou}, \bibinfo{person}{Jun Shirako}, {and} \bibinfo{person}{Vivek Sarkar}.} \bibinfo{year}{2024}\natexlab{}.
\newblock \showarticletitle{APPy: Annotated Parallelism for Python on GPUs}. In \bibinfo{booktitle}{\emph{Proceedings of the 33rd ACM SIGPLAN International Conference on Compiler Construction}}. \bibinfo{pages}{113--125}.
\newblock


\end{thebibliography}
